\documentclass[sigconf,nonacm]{acmart}

\usepackage{enumitem}
\usepackage{multirow}
\usepackage{subcaption}
\usepackage{algorithm}
\usepackage{algpseudocode}
\usepackage{placeins}
\usepackage{comment}
\usepackage{balance}

\usepackage{tikz}
\usetikzlibrary{positioning}

\usepackage{cleveref}

\newif\ifLongversion
\Longversionfalse
\newif\ifappendix
\appendixfalse

\newcommand{\ComplexityDijkstraSingle}{%
  O\!\left((|E|+N)\log N\right)%
}
\newcommand{\ComplexityDijkstraAllRanks}{%
  O\!\left(n(|E|+N)\log N\right)%
}
\newcommand{\ComplexityOracle}{%
  O\!\left(n(|E|+N)\log N+n^2\right)%
}
\newcommand{\ComplexityNiter}{%
  O\!\left(\frac{|E|\log|E|}{\epsilon^2}\right)%
}
\newcommand{\ComplexityTotal}{%
  O\!\left(
    \frac{|E|\log|E|}{\epsilon^2}
    \left[
      n(|E|+N)\log N+n^2
    \right]
  \right)%
}

\newcommand{\ComplexityAllgather}{%
  O\!\left(
    \frac{|E|\log|E|}{\epsilon^2}
    \left[\, n\bigl(N+|E|\bigr)\log N+n^3 \,\right]
  \right)%
}
\newcommand{\ComplexityMultiSet}{%
  O\!\left(
    \frac{|E|\log|E|}{\epsilon^2}
    \left[
      n(|E|+N)\log N+k \cdot n^2
    \right]
  \right)%
}
\newcommand{\ComplexityForestColl}{%
  O\!\left(n^8\right)%
}
\newcommand{\ComplexityProjectionDPSingle}{%
  O\!\left(N+|E|\right)%
}
\newcommand{\ComplexityProjectionDPAll}{%
  O\!\left(n\bigl(N+|E|\bigr)\right)%
}
\newcommand{\ComplexityPSArborescences}{%
  O\!\left(n^3\right)%
}
\newcommand{\ComplexityOraclePS}{%
  O\!\left(n\bigl(N+|E|\bigr)+n^3\right)%
}
\newcommand{\ComplexityTotalPS}{%
  O\!\left(
    \frac{|E|\log|E|}{\epsilon^2}
    \left[\, n\bigl(N+|E|\bigr)+n^3 \,\right]
  \right)%
}

\newcommand{\ComplexityEndToEndPS}{%
  O\!\left(
    \frac{|E|\log|E|}{\epsilon^2}
    \left[\, n\bigl(N+|E|\bigr)+n^3 \,\right]
    + \frac{|E|\log|E|\,\log(|E|/\epsilon)}{\epsilon^3}
  \right)%
}

\begin{document}

\title{DualColl: Throughput-Optimal Bandwidth-Bound Collective
Communications via Fractional Tree Packing}

\author{Namir Ballan}
\affiliation{%
  \institution{Technion -- Israel Institute of Technology}
  \city{Haifa}
  \country{Israel}}
\email{namir.ballan@campus.technion.ac.il}

\author{Reuven Cohen}
\affiliation{%
  \institution{Technion -- Israel Institute of Technology}
  \city{Haifa}
  \country{Israel}}
\email{rcohen@cs.technion.ac.il}

\begin{abstract}
Large-scale distributed training increasingly makes collective
communication bandwidth-bound. In modern GPU clusters, communication
libraries such as NCCL implement collectives including AllReduce over
heterogeneous interconnects such as NVLink, InfiniBand, and Ethernet,
making efficient use of available network bandwidth a central systems
challenge.

ForestColl showed that collective communication on such fabrics can be
formulated using fractional tree packing; its scalable polynomial-time
construction is developed for AllGather, with AllReduce obtained by
composing AllGather and ReduceScatter. We present DualColl, which
targets bandwidth-bound AllReduce directly. Under the tree-based
AllReduce model, DualColl applies a Garg--K\"onemann-style
multiplicative-weights algorithm and, for any $0<\epsilon<1$, computes
a feasible schedule whose throughput is at least
$(1-\epsilon)^2$ of the optimal AllReduce throughput.
Its key component is an exact minimum-cost collective-tree oracle that
projects physical-link dual prices onto a logical graph of participating
ranks, reducing each iteration to shortest-path computations followed
by a minimum spanning tree.

We further extend the framework to local equal-next-hop packet-sprayed
routing and to networks without explicit per-tree bandwidth control.
For the latter, we introduce a tree-coupled max-min bandwidth-sharing
model and show that tree duplication approximates an explicit fractional
allocation with a provable accuracy-complexity tradeoff. Finally, we
extend the fractional-packing framework to additional collective
operations and multiple concurrent collective groups.
\end{abstract}


\maketitle

\section{Introduction}

Large-scale distributed training relies heavily on collective
communication. In data-parallel training, for example, workers compute
gradients locally and then use AllReduce to obtain the same aggregated
result. Similar communication patterns arise in tensor, pipeline, and
expert parallelism, as well as distributed inference. As model sizes,
accelerator counts, and local computation rates grow, these collectives
increasingly become bandwidth-bound, making efficient use of network
capacity a central performance concern.

In GPU clusters, collective operations are commonly implemented by
communication libraries such as NVIDIA NCCL, which provides
topology-aware implementations of AllReduce, AllGather, ReduceScatter,
and related operations across multi-GPU and multi-node
systems~\cite{nccl}. Such systems combine heterogeneous resources:
intra-node links such as NVLink or NVSwitch may be substantially faster
than inter-node Ethernet or InfiniBand links, and oversubscription may
vary across the network hierarchy. Consequently, obtaining high
collective throughput requires communication schedules that account for
the capacities and topology of the underlying fabric.

A useful abstraction for bandwidth-bound collective communication is a
collection of communication trees. For AllReduce, a tree reduces a data
chunk toward a root and broadcasts the result back along the tree.
Different chunks may be assigned to different trees, each operating at
a continuous rate. With data partitioned proportionally to these rates,
the resulting steady-state collective throughput is the sum of the
bandwidths assigned to the active trees.

ForestColl~\cite{forestcoll} established fractional tree packing as a
powerful approach to collective scheduling on heterogeneous fabrics. In
this formulation, candidate trees receive continuous rates whose sum is
maximized subject to physical-link capacity constraints. ForestColl's
scalable polynomial-time construction is developed for AllGather, with
AllReduce obtained by composing AllGather with ReduceScatter. This
composition, however, is not guaranteed to be optimal for AllReduce, as
the counterexample in Appendix~\ref{app:counterexample} shows.

DualColl instead targets bandwidth-bound AllReduce directly. It retains
the fractional tree-packing abstraction but develops a different
algorithmic framework for solving it.
For the tree-based AllReduce formulation, DualColl computes a feasible
solution whose throughput is within a $(1-\epsilon)^2$ factor of the
optimum. To the best of our knowledge, DualColl is the first polynomial-time
algorithm with a provable near-optimality guarantee for bandwidth-bound
AllReduce on switched topologies.

DualColl uses a Garg--K\"onemann-style multiplicative-weights algorithm.
The algorithm maintains a dynamic price for every physical link and
repeatedly requests a minimum-cost collective tree. Its key
problem-specific component is an exact oracle that projects physical-link
prices onto a logical graph containing only the participating ranks and
then computes a minimum spanning tree. This avoids enumerating the
exponentially many candidate trees. 

We also extend the framework toward practical network operation.
First, we consider local equal-next-hop packet spraying, in which a
logical flow is distributed fractionally over multiple shortest-path
links rather than routed along a scheduler-selected path. Fine-grained
multipath packet spraying is increasingly relevant in AI-oriented
network fabrics; for example, the Ultra Ethernet Transport architecture
supports packet spraying across multiple paths to improve load balancing
and network utilization for AI and HPC workloads~\cite{uet}. Our model
captures the topology-based, congestion-oblivious form of such spraying.

Second, because real networks do not generally expose explicit per-tree
bandwidth control, we introduce a tree-coupled max-min bandwidth-sharing
model. We show that duplicating trees acts as an effective weight under
max-min sharing and can approximate a desired fractional allocation with
a provable accuracy--complexity tradeoff. Finally, we extend the
fractional-packing framework to additional collective operations and
multiple concurrent collective groups.

The main contributions of this paper are as follows:
\begin{itemize}[
    leftmargin=*,
    labelsep=0.5em,
    topsep=0pt,
    partopsep=0pt,
    itemsep=0.3em,
    parsep=0pt
]
    \item We give, to the best of our knowledge, the first
    polynomial-time algorithm with a provable near-optimality guarantee
    for bandwidth-bound AllReduce on switched topologies. Unlike
    ForestColl's composition of AllGather and ReduceScatter, which can
    be strictly suboptimal for AllReduce
    (Appendix~\ref{app:counterexample}), DualColl targets the AllReduce
    formulation directly.

    \item We develop a Garg--K\"onemann-style approximation algorithm
    with an exact minimum-cost collective-tree oracle based on dual-cost
    projection and a minimum spanning tree, avoiding explicit
    enumeration of exponentially many candidate trees.

    \item We extend the model to local equal-next-hop packet spraying,
    where a logical flow consumes fractional capacity on multiple
    physical links.

    \item We introduce a tree-coupled max-min bandwidth-allocation model
    and a tree-duplication construction that approximates explicit
    fractional tree bandwidths without requiring direct per-tree
    bandwidth control.

    \item We extend the framework to multiple concurrent collective
    groups and other collective operations. In particular, replacing
    the AllReduce oracle with one based on rooted spanning
    arborescences enables AllGather within the same fractional-packing
    framework.
\end{itemize}

The remainder of the paper is organized as follows.
Section~\ref{sec:model} presents the network model and the fractional
tree-packing formulation, and Section~\ref{sec:tree-optimality} proves
that restricting AllReduce schedules to trees loses no throughput.
Section~\ref{sec:dualcoll-algorithm} develops the DualColl approximation
algorithm, its minimum-cost collective-tree oracle, and its approximation
and complexity guarantees. Section~\ref{sec:packet-spraying-mmf} considers
packet-sprayed routing and tree-coupled max-min realization.
Section~\ref{sec:other-collectives} extends the framework to AllGather,
ReduceScatter, and multiple concurrent AllReduce groups.
Section~\ref{sec:related-work} reviews related work, and
Section~\ref{sec:conclusion} concludes the paper.
\section{Model and ForestColl Foundation}
\label{sec:model}

This section introduces the network and collective-communication model
used by DualColl. We then restate the fractional tree-packing
formulation underlying ForestColl, which serves as the starting point
for our work. Throughout this section, we consider the baseline
scheduler-routed model in which every logical communication edge is
mapped to a single chosen physical path. The extensions considered later are
introduced separately in Section~\ref{sec:packet-spraying-mmf}.

\subsection{Physical Network Model}
\label{sec:network-model}

We model the physical network as a directed graph
\[
G=(V\cup S,E),
\]
where \(V\) is the set of compute nodes, or ranks, and \(S\) is the set
of switches. Each directed physical edge \(e\in E\) has capacity
\(C_e>0\).

A full-duplex physical link is represented by two directed edges,
one in each direction, and the two directions are accounted for
by separate capacity constraints. In this paper, we
assume that the two directed edges corresponding to the same
physical link have equal capacity.

For a given collective operation, let \(R\subseteq V\) denote the set
of participating ranks. We initially consider a single collective group
and assume \(R=V\) for notational simplicity. Multiple concurrent
collective groups are considered in Section~\ref{sec:other-collectives}.

\subsection{Tree-Based AllReduce}
\label{sec:tree-allreduce}

In an AllReduce operation, every participating rank initially
holds an input array, and every rank must obtain the element-wise
reduction of all input arrays.

A tree-based AllReduce schedule performs this operation in two
phases. First, data is reduced from the leaves of a rooted logical tree
toward its root. Intermediate ranks combine the values received from
their children with their own local values. Second, the root broadcasts
the reduced result back to all ranks along the tree. Note that because data is combined at each intermediate node, the amount of data that a node sends does not increase as we go higher in the tree.

A logical tree contains only participating compute nodes. Switches are
not represented as logical vertices; they appear only in the physical
network over which the logical edges are routed.

Let \(T\) be a logical spanning tree over \(R\). Each logical edge of
\(T\) is mapped by the scheduler to a single physical path in \(G\).
The complete routing of the logical edges determines how much capacity
the tree consumes on every physical edge.

A schedule may use several logical trees in parallel. The input array
is divided into disjoint chunks, and each chunk is assigned to one
tree. Let
\[
\mathcal{T}=\{T_1,T_2,\ldots,T_k\}
\]
be the set of selected trees, and let \(\lambda_T\geq 0\) denote the
bandwidth assigned to tree \(T\), per direction, so that the reduce and
broadcast phases each operate simultaneously at rate \(\lambda_T\).

Restricting attention to schedules of this tree form loses no throughput.
Intuitively, consider a tiny shard of the data which is never split by some optimal schedule. If two ranks
both act as roots for that shard---each independently computing the fully
reduced value---then work is wasted: the two roots must communicate their
own contributions to one another in any case, so rather than exchanging
raw contributions it is always at least as good for one of them to send
the other its fully reduced result, eliminating one of the two roots.
Iterating this observation collapses every shard to a single root, after
which its reduction and broadcast each form a tree; a further exchange
argument shows the reduce and broadcast trees may be taken identical. We
state the resulting guarantee here and prove it in
Section~\ref{sec:tree-optimality}.

\begin{theorem}[Tree-based AllReduce is optimal]
\label{thm:tree-optimal-body}
Under scheduler-routed single-path routing, on a full-duplex network in
which every physical link has equal capacity in both directions, there exists a bandwidth-optimal AllReduce schedule in
which each shard is reduced along a tree to a single root and broadcast
from that root along the same tree.
\end{theorem}

\subsection{Physical-Link Incidence}
\label{sec:incidence}
For every physical edge \(e\in E\) (a single directed edge) and logical
tree \(T\), define \(A_{e,T}\) as the amount of capacity consumed on
\(e\) when \(T\) operates at one unit of bandwidth in both the reduce and broadcast directions
simultaneously.

Under the scheduler-routed single-path model considered in this
section, \(A_{e,T}\) is the number of routed logical tree edges whose
physical paths traverse \(e\). It may be greater than one if several
logical edges of the same tree use the same physical edge.

If tree \(T\) is assigned bandwidth \(\lambda_T\), then its capacity
consumption on physical edge \(e\) is
\[
A_{e,T}\lambda_T.
\]

Consequently, the aggregate bandwidth assigned to all trees must
satisfy
\[
\sum_{T\in\mathcal{T}} A_{e,T}\lambda_T \leq C_e,
\qquad \forall e\in E.
\]

\subsection{ForestColl Tree-Packing Formulation}
\label{sec:forestcoll-formulation}

ForestColl established tree packing as a
throughput-optimal abstraction for bandwidth-bound collective
communication. DualColl adopts this abstraction as its base
optimization problem.

Let \(\mathcal{S}\) denote the set of all feasible logical spanning
trees over the participating ranks, including their physical-path
mappings. A continuous variable \(\lambda_T\) specifies the bandwidth
assigned to every candidate tree \(T\in\mathcal{S}\).

The fractional tree-packing problem is
\begin{align}
\max_{\lambda}\quad
    & \sum_{T\in\mathcal{S}} \lambda_T
    \label{eq:fractional-packing-objective}\\
\text{subject to}\quad
    & \sum_{T\in\mathcal{S}} A_{e,T}\lambda_T
      \leq C_e,
      && \forall e\in E,
    \label{eq:fractional-packing-capacity}\\
    & \lambda_T\geq 0,
      && \forall T\in\mathcal{S}.
    \label{eq:fractional-packing-nonnegative}
\end{align}

The objective maximizes the aggregate bandwidth of the selected trees,
while the constraints ensure that their combined physical-link
consumption does not exceed any link capacity. We write
\begin{equation*}
\mathrm{OPT} = \max_{\lambda}\ \Bigl\{\, \textstyle\sum_{T\in\mathcal{S}} \lambda_T \ : \ \lambda \text{ satisfies Equations~\eqref{eq:fractional-packing-capacity}--\eqref{eq:fractional-packing-nonnegative}} \,\Bigr\}
\end{equation*}
for the optimal value of this program, i.e., the maximum aggregate tree
bandwidth achievable under the model. This fractional tree-packing
formulation is throughput-optimal in the sense of ForestColl: its optimum
$\mathrm{OPT}$ is the largest collective throughput attainable by any
schedule in the model. DualColl does not change this formulation; it
computes a solution whose throughput is provably close to $\mathrm{OPT}$,
as made precise in Section~\ref{sec:approximation-complexity}.

Although the formulation contains one variable for every possible
tree, only trees with \(\lambda_T>0\) are instantiated in the resulting
schedule. Because the number of capacity constraints is polynomial in
the network size, there exists an optimal solution that uses only a polynomial
number of trees, even though the candidate set \(\mathcal{S}\) may be
exponentially large.

\subsection{Throughput and Data Partitioning}
\label{sec:data-partitioning}

We next explain why the objective in
Equation~\eqref{eq:fractional-packing-objective} corresponds to
maximizing the throughput of a bandwidth-bound AllReduce.

Suppose each rank starts with $X$ units of data and receives a total of \(X\) units of reduced data;
in AllReduce every rank obtains the full reduced result, so \(X\) is the
size of the data each rank receives, not the aggregate data summed over
all ranks. Let \(x_T\) be the amount of this data carried by tree \(T\),
so that
\[
\sum_{T\in\mathcal{T}}x_T=X.
\]

A tree assigned bandwidth \(\lambda_T\) completes its chunk in time
\[
\tau_T=\frac{x_T}{\lambda_T}.
\]

Although a tree-based AllReduce consists of a reduction phase
followed by a broadcast phase, we adopt a bandwidth-dominated,
pipelined steady-state interpretation of this expression. We assume the
messages are sufficiently large and finely pipelined that the reduction
and broadcast of different data segments overlap in steady state, and we
define \(\lambda_T\) as the sustainable steady-state throughput of tree
\(T\) under this pipelined execution. Startup latency and pipeline
fill/drain effects are ignored. Under this interpretation the two phases
do not serialize per segment, and \(\tau_T=x_T/\lambda_T\) is the
steady-state completion time of the chunk assigned to \(T\).

Because all trees execute in parallel, the collective completion time
is determined by the slowest tree. It is minimized when all active
trees complete simultaneously. This is achieved by assigning data in
proportion to the tree bandwidths:
\[
x_T
=
X
\frac{\lambda_T}
     {\sum_{T'\in\mathcal{T}}\lambda_{T'}}.
\]

Under this allocation, every active tree completes after
\[
\tau
=
\frac{X}
     {\sum_{T\in\mathcal{T}}\lambda_T}.
\]
We have just shown that this completion time is achievable. The
following lemma gives a matching lower bound, so this partition is the
best possible.

\begin{lemma}[Completion-time lower bound]
\label{lem:partition-lower-bound}
Let $\mathcal{T}$ be a set of trees with bandwidths $\lambda_T>0$ and
aggregate bandwidth $\Lambda=\sum_{T\in\mathcal{T}}\lambda_T$, and let
the $X$ units of data be partitioned among them in any way. Then some tree
completes no earlier than $X/\Lambda$.
\end{lemma}

\begin{proof}
Let tree $T$ carry $x_T\ge0$ units, with $\sum_{T}x_T=X$, and suppose for
contradiction that every tree completes before time $X/\Lambda$. Then every tree $T$ completes at time
$x_T/\lambda_T<X/\Lambda$, that is, $x_T<X\lambda_T/\Lambda$. Summing over all trees,
\[
X=\sum_{T\in\mathcal{T}}x_T<\frac{X}{\Lambda}\sum_{T\in\mathcal{T}}\lambda_T=X,
\]
a contradiction.
\end{proof}

Therefore, minimizing the completion time of a bandwidth-bound
AllReduce is equivalent to maximizing the aggregate tree
bandwidth:
\[
\sum_{T\in\mathcal{T}}\lambda_T.
\]

This equivalence relies on the steady-state interpretation above, under
which the trees operate concurrently at their assigned sustainable rates.

\section{Optimality of Tree-Based AllReduce}
\label{sec:tree-optimality}

This section proves Theorem~\ref{thm:tree-optimal-body}. The model of
Section~\ref{sec:model} restricts AllReduce schedules to collections of
logical trees, each of which reduces a shard of the data toward its root
and then broadcasts the reduced shard back along the same tree. We show
that this restriction loses no throughput, so the optimum $\mathrm{OPT}$
of the fractional tree-packing formulation is the bandwidth-optimal
AllReduce throughput itself, and not merely the optimum over a restricted
class of schedules. Starting from an arbitrary bandwidth-optimal schedule,
we normalize it in three steps, none of which increases any per-edge load
or the completion time: each shard is first reduced to a single root
(Lemma~\ref{lem:single-root}), its reduce and broadcast phases are then
made trees rooted there (Lemma~\ref{lem:phase-trees}), and finally, under
single-path routing, the two trees are made to coincide
(Lemma~\ref{lem:same-tree}). Lemmas~\ref{lem:promote}--\ref{lem:not-slower}
are the tools used along the way.

\paragraph{Setting.}
We fix a \emph{shard}: a unit of data that the schedule
treats uniformly, so that the shard is never split into smaller pieces and
every rank's contribution to it is moved and reduced as a single whole.
A schedule for the shard is a finite set of \emph{flows}, each carrying a
payload from one rank to another, and each assigned a rate, that is, a
bandwidth. The payload of every flow is the reduction of the contributions
of some subset of the ranks, and the fully reduced shard is the reduction
over all ranks; since reduction is applied elementwise, every payload has
the same size as the shard. A flow of rate $b$ places a load of $b$ times the fraction of
its traffic routed over directed edge $e$ on $e$; the load on $e$ is the
sum of these over all flows of all shards, and the schedule is feasible if
the load on every directed edge $e$ is at most $C_e$. As in
Section~\ref{sec:model}, every physical link has equal capacity in both
directions, $C_e = C_{\bar e}$, where $\bar e$ denotes the reverse of
directed edge $e$.

Flow $g$ \emph{depends} on flow $f$ if the payload of $g$ is computed from
data delivered by $f$; the dependency relation is acyclic. In the
bandwidth-bound model, a flow $g$ that depends on $f$ does not wait for
$f$ to end: it starts as soon as the data of $f$ begins to arrive and
pipes that data onward as it moves, so dependent flows run concurrently.
We call a rank a \emph{root} for the
shard if it obtains the fully reduced shard without receiving it, that is,
it computes the final value locally from partial values it has received.
At least one root exists, since the first rank to hold the fully reduced
shard cannot have received it, and every rank ends holding the fully
reduced shard. For a single directed phase $T$ (a reduce in-tree or a
broadcast out-tree, together with its physical routing) we write $a_e(T)$
for the load $T$ places on directed edge $e$ per unit rate.

\begin{lemma}[Promotion to the fully reduced shard]
\label{lem:promote}
Let $f$ be a flow whose source already possesses the fully reduced shard.
If the payload of $f$ is replaced by the fully reduced shard, and the same
replacement is applied to every flow that depends on $f$, directly or
transitively, then the resulting schedule is still legal.
\end{lemma}

\begin{proof}
Consider the set $\mathcal{D}$ consisting of $f$ together with all flows
that depend on $f$ directly or transitively. We promote every flow in
$\mathcal{D}$ to carry the fully reduced shard. We must check that each
promoted flow's payload is actually available at its source when the flow
is sent, so that the schedule remains executable.

Process the flows of $\mathcal{D}$ in an order consistent with the
dependency relation, which is possible since that relation is acyclic. The
first flow is $f$ itself, and its source possesses the fully reduced shard
by hypothesis, so it can send it. Consider any later flow $g\in\mathcal{D}$
with source $u$. Every flow of $\mathcal{D}$ on which $g$ depends has been
processed earlier and now delivers the fully reduced shard, so $u$ is receiving the fully reduced shard from flows that $g$ depends on and can therefore send
it. Thus every promoted flow is executable.

A rank that receives a promoted flow receives the fully reduced shard, and
holding the final value it is complete for this shard regardless of any
other partial values, so its state is consistent. A flow \emph{not} in
$\mathcal{D}$ does not depend on $f$, so none of its inputs were altered
and its payload is unchanged; hence every rank outside the recipients of
$\mathcal{D}$ ends exactly as before, and in particular every rank that
was already complete remains complete. The schedule is therefore legal.
\end{proof}

\begin{lemma}[Equal rates within a shard]
\label{lem:equal-rates}
Without loss of generality, all flows belonging to the same shard are
assigned the same rate.
\end{lemma}

\begin{proof}
We first discard every flow that does not contribute, directly or
transitively, to a final reduced value held by some rank; removing such a
flow changes what no rank computes and can only reduce the load on every
link. Let the shard have size $s$, let its remaining flows have rates
$b_g>0$, and let
$b_{\min}=\min_g b_g$. Every flow of the shard carries a payload of size
$s$. Flow $g$ therefore needs time $s/b_g$ to transmit its payload, and the
shard cannot complete before every flow has done so, that is, before time
$\max_g s/b_g = s/b_{\min}$. Conversely, since a flow pipes its payload
onward as its inputs arrive, its effective rate is the minimum of its own
rate and the effective rates of the flows it depends on, so every flow
progresses at effective rate at least $b_{\min}$. Ignoring start-up and
fill/drain effects, as the model of Section~\ref{sec:model} does, every
flow finishes by time $s/b_{\min}$, and the shard completes at exactly
$s/b_{\min}$.

Now lower the rate of every flow $g$ with $b_g>b_{\min}$ to $b_{\min}$.
All flows of the shard then have rate $b_{\min}$, so by the argument above
the shard still completes at time $s/b_{\min}$. The load that $g$ places on
a directed edge is proportional to its rate, so every per-edge load is
unchanged or smaller; every capacity constraint that held before still
holds, and the remaining shards are unaffected.
\end{proof}

\begin{lemma}[Promotion is not slower]
\label{lem:not-slower}
Reordering flows by the promotion of Lemma~\ref{lem:promote}, so that a
promoted flow forwards the fully reduced shard as it arrives at its
source, does not make the schedule slower.
\end{lemma}

\begin{proof}
By Lemma~\ref{lem:equal-rates} all flows of the shard have the same rate
$b$, so the shard completes at time $s/b$, where $s$ is its size. The
promotion keeps every flow with its route and its rate, and changes only
payloads, each of which still has size $s$. Since all rates are equal, no
flow has a higher rate than a flow it depends on, including the flows that
now deliver the fully reduced shard to the source of a promoted flow; so
no flow is throttled, and the shard still completes at time $s/b$. Finally, the bandwidth consumed
on every link is unchanged, so the rest of the collective is unchanged and
finishes in the same time.
\end{proof}

\begin{lemma}[Single root]
\label{lem:single-root}
For every feasible schedule there is a feasible schedule with no larger
per-edge loads and no larger completion time in which each shard has
exactly one root. Equivalently, without loss of generality each shard is reduced
into a single root $r$ and every rank other than $r$ receives the fully
reduced shard at some point of the collective.
\end{lemma}

\begin{proof}
By Lemma~\ref{lem:equal-rates} we may assume that all flows of the shard
have the same rate $b$; this does not increase any per-edge load and does
not change the completion time. Suppose the shard has $k\ge 2$ roots
$r_1,\ldots,r_k$. Each rank other
than a root receives the fully reduced shard from some root; let $S_i$ be
the set of ranks to which $r_i$ is responsible for delivering it, with
$r_i\in S_i$. We may assume the $S_i$ form a partition of the ranks: if a
rank were served by two roots, one of the two deliveries is redundant and
may be deleted, which only reduces load, and reassigning the rank to a
single server leaves the sets disjoint.

Call a flow \emph{inter-group} if its endpoints lie in different sets of
this partition. At least one inter-group flow exists: each root holds the
\emph{fully} reduced shard, which incorporates the contributions of ranks
in every group, and those contributions can only leave their group along
an inter-group flow. Since the dependency relation is acyclic and there
are finitely many flows, we may choose a \emph{last} inter-group flow
$f$, one on which no other inter-group flow depends (there may be several). Say $f$ goes from a
rank of $S_i$ to a rank of $S_j$, with $i\neq j$.

Consider the flows that depend on $f$, directly or transitively, and that
lead to $r_j$. If no such chain exists, then $f$ delivers data that never
reaches $r_j$ and never leaves $S_j$ toward any other group, so $f$ is
useless and may be deleted, decreasing the load and leaving the set of
roots unchanged; we may therefore assume such a chain exists. Let
$\mathcal{C}$ be this chain of flows, from $f$ up to and including the
flow that delivers into $r_j$.

We now promote $\mathcal{C}$: replace the payload of $f$ by the fully
reduced shard and apply the same replacement along $\mathcal{C}$. For
this to be executable, the source of $f$ must hold the fully reduced shard
when $f$ fires. It lies in $S_i$, so $r_i$ delivers the fully reduced
shard to it; we therefore let $f$ forward the fully reduced shard as it
arrives at its source. This reordering is legitimate because, as shown
below, no flow of $\mathcal{C}$ carries data that $r_i$ or the source of
$f$ needs in order to obtain the fully reduced shard, so it does not
prevent either of them from completing. By Lemma~\ref{lem:promote} the
result is a legal schedule, and by Lemma~\ref{lem:not-slower} it is not
slower.

It remains to check that this promotion does not destroy the very data it
relies on, that is, that no flow of $\mathcal{C}$ was needed for $r_i$ to
obtain the fully reduced shard. This is where the choice of $f$ as the
\emph{last} inter-group flow is used. Every flow of $\mathcal{C}$ depends
on $f$, hence carries data that has entered $S_j$ through $f$. For such
data to contribute to $r_i$'s reduction it would have to return from
$S_j$ to $S_i$, which requires an inter-group flow depending on $f$; no
such flow exists, by the maximality of $f$. Hence $r_i$ still computes the
fully reduced shard exactly as before, and remains a root.

In the modified schedule $r_j$ receives the fully reduced shard along the
promoted chain instead of computing it from partial values, so $r_j$ is no
longer a root, and it still delivers the shard to $S_j$ as before. The
number of roots has decreased by one, while the per-edge loads and the
completion time are unchanged. Repeating the argument until a single root
remains, and doing so for every shard, yields the claimed schedule.
\end{proof}

\begin{lemma}[Each phase is a tree]
\label{lem:phase-trees}
Consider a shard with a single root $r$. Without loss of generality, its
reduce phase is an in-tree, that is, a spanning in-arborescence rooted at
$r$, and its broadcast phase is an out-tree, a spanning out-arborescence
rooted at $r$.
\end{lemma}

\begin{proof}
By Lemma~\ref{lem:equal-rates} all flows of the shard have a common rate
$b$, and the shard completes at time $s/b$, where $s$ is its size.

\emph{Reduce phase.} Call a flow a \emph{reduce flow} if the computation of
the fully reduced shard at $r$ depends on it, directly or transitively. Let
$H$ be the \emph{send graph} of the reduce phase: its vertices are the
ranks, and it has a directed edge $(u,w)$ whenever at least one reduce flow
goes from $u$ to $w$. The graph $H$ need not be an in-arborescence; a rank
may, for instance, send to several ranks. However, every rank has a
directed path to $r$ in $H$: the contribution of a rank $v$ is part of the
fully reduced shard, so it leaves $v$ on a reduce flow and reaches $r$
along a chain of reduce flows, each of which starts at the rank where the
previous one ends.

We select an in-arborescence inside $H$. Let $d(v)$ be the length of a
shortest directed path from $v$ to $r$ in $H$, and consider the ranks in
increasing order of $d(v)$. For every rank $v\ne r$, choose an arbitrary
shortest path from $v$ to $r$ and keep only its first edge $(v,w)$, so
that $d(w)=d(v)-1$. In the resulting graph $R$, every rank other than $r$
has exactly one outgoing edge and $r$ has none, and $d$ strictly decreases
along every edge, so $R$ contains no cycle. Hence $R$ is an
in-arborescence spanning the ranks and rooted at $r$.

We now reduce along $R$: every rank $v$ sends to its parent the reduction
of its own contribution and of the payloads it receives from its children,
that is, the reduction over the ranks of its subtree in $R$. The subtrees
of the children of $r$ are disjoint and, together with $r$, cover all
ranks, so $r$ obtains the fully reduced shard. Each edge $(v,w)$ of $R$ is
carried by one of the original reduce flows from $v$ to $w$, which keeps
its physical route and its rate $b$; only its payload changes, to the
reduction over the subtree of $v$, which is again the reduction of a
subset of the ranks and therefore has size $s$. All other reduce flows
are dropped.

\emph{Broadcast phase.} Let $H'$ be the send graph of the flows that carry
the fully reduced shard, with a directed edge $(u,w)$ whenever at least one
such flow goes from $u$ to $w$. By Lemma~\ref{lem:single-root}, every rank
other than $r$ receives the fully reduced shard from a rank that already
holds it; following these deliveries backward, which terminates because
the dependency relation is acyclic, leads to $r$, so every rank is
reachable from $r$ in $H'$. Symmetrically to the reduce phase, let $d'(v)$
be the length of a shortest directed path from $r$ to $v$ in $H'$, and for
every rank $v\ne r$ keep only the last edge $(u,v)$ of an arbitrary
shortest path from $r$ to $v$, so that $d'(u)=d'(v)-1$. The result is an
out-arborescence $D$ spanning the ranks and rooted at $r$, along which $r$
forwards the fully reduced shard; each edge of $D$ is carried by one of
the original flows, with its route, rate and payload unchanged. All other
flows carrying the fully reduced shard are dropped.

\emph{Cost.} The new schedule for the shard uses a subset of its original
flows, each with its original route, the common rate $b$, and a payload of
size $s$. Hence the load on every directed edge is unchanged or smaller,
and, since all flows still have rate $b$, the shard still completes at time
$s/b$, as in the proof of Lemma~\ref{lem:equal-rates}. The flows of the
other shards are untouched, so the rest of the collective is unaffected.
\end{proof}

We refer to a tree that reduces along $T$ toward
its root and broadcasts along $T$ reversed as a \emph{simple} tree; this
is precisely the tree of Section~\ref{sec:model}, and by construction a
simple tree built on a directed phase $T$ loads directed edge $e$ with
$a_e(T)+a_{\bar e}(T)$ per unit rate, and loads $\bar e$ with the same
amount.

\begin{lemma}[Common reduce and broadcast tree]
\label{lem:same-tree}
Under scheduler-routed single-path routing, consider a feasible schedule in which each
shard $i\in\{1,\ldots,k\}$ is reduced along an in-tree $R_i$ to root $r_i$
and broadcast from $r_i$ along an out-tree $D_i$, at rate $\lambda_i$.
Then there is a feasible schedule consisting of $2k$ simple trees in which
every shard completes at the same time.
\end{lemma}

\begin{proof}
For each $i$ let $T_i$ be the simple tree built on $R_i$ (reduce along
$R_i$, broadcast along $R_i$ reversed) and let $T_{k+i}$ be the simple
tree built on $D_i$ (reduce along $D_i$ reversed, broadcast along $D_i$),
both rooted at $r_i$. Assign each of $T_i$ and $T_{k+i}$ the rate
$\lambda_i/2$, and split shard $i$ into two halves, one carried by each.
Each half is fully reduced to $r_i$ and broadcast from $r_i$, so the new
schedule performs AllReduce. If shard $i$ has size $s_i$, each half has
size $s_i/2$ and rate $\lambda_i/2$, so it completes at time
$(s_i/2)/(\lambda_i/2)=s_i/\lambda_i$, exactly when shard $i$ completed
before.

It remains to show that no directed edge is overloaded. The proof of this relies on single-path routing, under which reversing a
routed logical edge places on $\bar e$ exactly the load the original placed
on $e$. Fix a physical
link with directed edges $e$ and $\bar e$, and let
\begin{align*}
L_e &= \sum_{i}\lambda_i\bigl[a_e(R_i)+a_e(D_i)\bigr],\\
L_{\bar e} &= \sum_{i}\lambda_i\bigl[a_{\bar e}(R_i)+a_{\bar e}(D_i)\bigr]
\end{align*}
be the loads of the original schedule on $e$ and $\bar e$. Feasibility of
the original schedule gives $L_e\le C_e$ and $L_{\bar e}\le C_{\bar e}=C_e$.
In the new schedule, $T_i$ loads $e$ with $(\lambda_i/2)[a_e(R_i)+a_{\bar
e}(R_i)]$ and $T_{k+i}$ loads $e$ with $(\lambda_i/2)[a_e(D_i)+a_{\bar
e}(D_i)]$, so the total load on $e$ is
\[
\sum_i \frac{\lambda_i}{2}\Bigl[a_e(R_i)+a_{\bar e}(R_i)+a_e(D_i)+a_{\bar e}(D_i)\Bigr]
=\frac{L_e+L_{\bar e}}{2}.
\]
By the symmetry of simple trees the load on $\bar e$ is the same
quantity. Hence the load on each direction of the edge becomes the average
of the two original loads, and
\[
\frac{L_e+L_{\bar e}}{2}\;\le\;\max\{L_e,L_{\bar e}\}\;\le\;C_e .
\]
Thus splitting every tree evenly across the two directions can only lower
the larger of the two directional loads, and the new schedule respects
every capacity.
\end{proof}

\begin{theorem}[Tree-based AllReduce is optimal; restatement of Theorem~\ref{thm:tree-optimal-body}]
\label{thm:tree-structure-optimal}
Under scheduler-routed single-path routing, on a full-duplex network in
which every physical link has equal capacity in both directions, there exists a bandwidth-optimal AllReduce schedule in
which each shard is reduced along a tree to a single root and broadcast
from that root along the same tree.
\end{theorem}

\begin{proof}
\label{proof:tree-structure}
Consider a bandwidth-optimal AllReduce schedule.
Applying Lemma~\ref{lem:single-root} to every shard yields a feasible
schedule with no larger completion time in which each shard has a single
root; by Lemma~\ref{lem:phase-trees}, we may further assume that each
shard is reduced along an in-tree to its root and broadcast from it along
an out-tree. Applying
Lemma~\ref{lem:same-tree} yields a feasible schedule of simple trees in
which every shard completes at the same time. This schedule is therefore
a bandwidth-optimal AllReduce schedule consisting entirely of simple
trees.
\end{proof}

\section{DualColl's Algorithm}
\label{sec:dualcoll-algorithm}

Section~\ref{sec:model} formulated bandwidth-bound
AllReduce scheduling as the fractional tree-packing problem in
Equations~\eqref{eq:fractional-packing-objective}--%
\eqref{eq:fractional-packing-nonnegative}. Although this formulation
contains only one capacity constraint per physical link, it has one
variable for every feasible collective tree and its physical-path
mapping. The number of such routed trees is exponential in the number
of participating ranks, which prevents their direct enumeration.

DualColl solves this problem using a
Garg--K\"onemann-style multiplicative-weights framework. Rather than
constructing the complete linear program, the algorithm maintains a
positive price for every physical edge and repeatedly selects a
collective tree having minimum cost under the current prices. The
selected tree is added to the fractional packing, after which the
prices of the physical edges it uses are increased.

The key problem-specific component of the framework is therefore a
minimum-cost collective-tree oracle. DualColl implements this oracle by
projecting the current physical-edge prices onto a logical graph over
the participating ranks and computing a minimum spanning tree in the
projected graph.

\subsection{Multiplicative-Weights Framework}
\label{sec:garg-konemann-framework}

DualColl maintains a positive price \(y_e\) for every physical edge and
a fractional bandwidth \(\lambda_T\) for every tree selected so far.
Let \(m=|E|\) be the number of physical edges. The algorithm begins by
setting
\begin{equation*}
\gamma = (1+\epsilon)\bigl((1+\epsilon)m\bigr)^{-1/\epsilon},
\qquad
y_e = \frac{\gamma}{C_e}\ \ \text{for every } e\in E,
\end{equation*}
and \(\lambda_T=0\) for all trees. The scale \(\gamma\) is the standard
Garg--K\"onemann initialization for the packing
LP~\cite{garg2007faster}; it depends only on the number of
constraints \(m=|E|\) and the accuracy \(\epsilon\), and is chosen so that
the final scaled packing satisfies the capacity constraints while
retaining a \((1-\epsilon)^2\) fraction of the optimum.

At each iteration, the algorithm performs the following steps.

\begin{enumerate}
    \item Invoke the minimum-cost collective-tree oracle under the
    current prices \(y\), obtaining the tree \(T^\star\) minimizing
    \(\sum_{e} A_{e,T}\,y_e\).

    \item Compute the bottleneck augmentation, the
    largest bandwidth \(T^\star\) may carry before one of its edges
    reaches unit relative load.
    \[\Delta = \min_{e:\,A_{e,T^\star}>0} C_e/A_{e,T^\star}\]

    \item Increase the bandwidth of \(T^\star\) by \(\Delta\).

    \item For every physical edge used by \(T^\star\), multiplicatively
    increase its price:
    \(y_e \gets y_e\bigl(1+\epsilon\, A_{e,T^\star}\Delta/C_e\bigr)\).

    \item Stop when the dual length
    \(D(y)=\sum_{e} C_e\,y_e\) reaches \(1\).
\end{enumerate}

The multiplicative updates discourage repeated use of highly loaded
physical edges. Initially, all physical-edge prices are small, and the
oracle selects a tree having low physical cost. As trees are added to
the packing, the prices of frequently used edges increase, so subsequent
oracle calls are encouraged to select trees that use less-congested
portions of the network.

Because each augmentation raises a tree's bandwidth without regard to the
edges already committed to earlier trees, the accumulated bandwidths may
exceed link capacities. At termination they are uniformly scaled down by
\(\log_{1+\epsilon}\bigl((1+\epsilon)/\gamma\bigr)\), yielding a feasible
packing
\(\hat\lambda_T=\lambda_T/\log_{1+\epsilon}\bigl((1+\epsilon)/\gamma\bigr)\)
that satisfies
Equation~\eqref{eq:fractional-packing-capacity}.

Algorithm~\ref{alg:dualcoll-framework} summarizes the framework.

\begin{algorithm}[h]
\caption{DualColl fractional tree-packing framework}
\label{alg:dualcoll-framework}
\begin{algorithmic}[1]
\Require Physical network \(G=(V\cup S,E)\), capacities
         \(\{C_e\}_{e\in E}\), participating ranks \(R\), and
         accuracy parameter \(0<\epsilon<1\)
\Ensure A feasible fractional tree packing \(\{\hat\lambda_T\}\)
\State \(\gamma \gets (1+\epsilon)\bigl((1+\epsilon)m\bigr)^{-1/\epsilon}\)
       \Comment{$m=|E|$}
\State \(y_e \gets \gamma / C_e\) for every \(e\in E\)
\State \(\lambda_T \gets 0\) for all trees
\While{\(D(y)=\sum_{e\in E} C_e\, y_e < 1\)}
    \State \(T^\star \gets \Call{MinCostCollectiveTree}{G,R,y}\)
    \State \(\Delta \gets \min_{e:\,A_{e,T^\star}>0}
            \dfrac{C_e}{A_{e,T^\star}}\)
           \Comment{bottleneck augmentation}
    \State \(\lambda_{T^\star} \gets \lambda_{T^\star}+\Delta\)
    \ForAll{\(e\in E\) with \(A_{e,T^\star}>0\)}
        \State \(y_e \gets y_e\Bigl(1+\epsilon\,
               \dfrac{A_{e,T^\star}\,\Delta}{C_e}\Bigr)\)
    \EndFor
\EndWhile
\State \(\hat\lambda_T \gets \lambda_T \big/
       \log_{1+\epsilon}\tfrac{1+\epsilon}{\gamma}\)
       for all \(T\)
\State \Return \(\{\hat\lambda_T\}\)
\end{algorithmic}
\end{algorithm}

The outer multiplicative-weights framework is generic to fractional
packing. The main network-specific challenge is implementing
Line~5 of Algorithm~\ref{alg:dualcoll-framework}, namely, finding a
minimum-cost collective tree efficiently.

\subsection{Minimum-Cost Collective-Tree Oracle}
\label{sec:collective-tree-oracle}

For a fixed price vector \(y\), the oracle must find a logical spanning
tree and its physical routing that minimize 
\[
\sum_{e \in E} y_e A_{e,T}.
\]
As in ForestColl, the scheduler controls the physical route of every
logical edge and may select any admissible path in the physical
network. Consequently, under the current dual prices, the minimum-cost
realization of a logical edge can be obtained through a shortest-path
computation using the dual prices as physical-edge weights.

For every ordered pair of participating ranks \(u,v\in R\), let
\[
d_y(u,v)
=
\min_{P\in\mathcal{P}(u,v)}
\sum_{e\in P}y_e,
\]
where \(\mathcal{P}(u,v)\) denotes the set of admissible directed
physical paths from \(u\) to \(v\).

An undirected logical tree edge \(\{u,v\}\) is used in one direction
during the reduction phase and in the reverse direction during the
broadcast phase. Its projected logical-edge weight is therefore
\begin{equation*}
    W_{\{u,v\}}
    =
    d_y(u,v)+d_y(v,u).
\end{equation*}

Thus, \(W_{\{u,v\}}\) is the total dual-weight cost of the logical tree
edge \(\{u,v\}\), summing the dual weights along the shortest paths used
in the two directions per unit of bandwidth assigned to that edge. Under
symmetric full-duplex links, the two directional costs may be equal, but
both directions remain represented.

DualColl constructs an undirected projected graph
\[
G_R=(R,E_R),
\]
whose vertices are the participating ranks. For every unordered pair
\(\{u,v\}\subseteq R\), the graph contains a logical edge with weight
\(W_{\{u,v\}}\). Physical switches do not appear as vertices of
\(G_R\); they are represented implicitly through the directed physical
paths used to compute the projected edge weights. Section~\ref{sec:proj}
describes how the projection weights are computed efficiently, both in
the scheduler-routed model considered here and in the packet-sprayed
model of Section~\ref{sec:packet-spraying-mmf}.

The oracle computes a minimum spanning tree of \(G_R\). For every
logical edge \(\{u,v\}\) selected by the minimum spanning tree, the
oracle records a minimum-cost physical path from \(u\) to \(v\) and a
minimum-cost physical path from \(v\) to \(u\). The resulting routed
logical tree is returned to the multiplicative-weights framework.

The correctness of this construction relies on the following property of
the routing model, which we state explicitly before using it.

\paragraph{Assumption (Route separability).}
Given a logical tree, the admissible physical route of each directed
logical edge may be selected independently of the routes assigned to all
other logical edges of the tree.

Route separability is what lets the minimum physical cost of a fixed
logical tree decompose into independent per-edge shortest-path
computations. It holds in the scheduler-routed model considered here,
where each logical edge is routed along a minimum-dual-cost physical path
chosen without reference to the other edges. Coupled path
constraints---for example, a requirement that the physical paths of
different logical edges be edge-disjoint---would violate separability and
require a different oracle, since the per-edge minimization would no
longer be independent.

The projection is valid because the physical dual cost of a routed
logical tree is additive over its logical edges.

Let \(T\) be a logical spanning tree. For every logical edge
\(\{u,v\}\in T\), let \(P_{u,v}\) denote the physical path used in one
direction and let \(P_{v,u}\) denote the physical path used in the
reverse direction. The dual cost of the routed tree is
\[
\operatorname{cost}_y(T)
=
\sum_{\{u,v\}\in T}
\left(
\sum_{e\in P_{u,v}}y_e
+
\sum_{e\in P_{v,u}}y_e
\right).
\]

Because the routes of different logical edges can be selected
independently (route separability, stated above), the minimum
physical cost of a fixed logical tree is
\begin{equation}
    \min_{\{P_{u,v},P_{v,u}\}}
    \operatorname{cost}_y(T)
    =
    \sum_{\{u,v\}\in T}W_{\{u,v\}}.
    \label{eq:tree-projected-cost}
\end{equation}

Minimizing the right-hand side of
Equation~\eqref{eq:tree-projected-cost} over all logical spanning trees
is exactly the minimum spanning tree problem in the projected graph
\(G_R\). Therefore, the projected MST and its associated directed
physical paths solve the oracle problem $\min_{T} \sum_{e \in E} y_e A_{e,T}$.

\begin{theorem}[Oracle correctness]
\label{thm:oracle-correctness}
Under the scheduler-routed model and the route-separability assumption above, dual-cost projection followed by a minimum spanning tree computation
returns a routed collective tree of minimum dual cost.
\end{theorem}

\begin{proof}[Proof of Theorem~\ref{thm:oracle-correctness}]
\label{proof:oracle-correctness}
For every logical edge \(\{u,v\}\), the projected weight
\(W_{\{u,v\}}\) is the minimum combined dual cost of realizing that
edge in both physical directions. By
Equation~\eqref{eq:tree-projected-cost}, the minimum routed cost of any
fixed logical tree equals the sum of its projected edge weights. A
minimum spanning tree minimizes this sum over all logical spanning
trees. Hence, the logical tree and physical paths returned by the
oracle minimize
\[
\sum_{e\in E}A_{e,T}y_e
\]
over all feasible routed collective trees.
\end{proof}

\subsection{Approximation and Time Complexity}
\label{sec:approximation-complexity}
\label{sec:oracle-complexity}

We first bound the cost of a single oracle invocation. Let
\(n=|R|\) be the number of participating ranks and let
\(N=|V\cup S|\) be the total number of physical-network vertices. The
oracle does two things: it computes the projection weights
\(W_{\{u,v\}}\), and it runs a minimum spanning tree on the projected
graph. Under the scheduler-routed model the projection weights are
shortest-path distances under the dual prices, so, using Dijkstra's
algorithm with a binary heap, one single-source computation requires
\(\ComplexityDijkstraSingle\) time and computing the distances from
every participating rank requires \(\ComplexityDijkstraAllRanks\).
Section~\ref{sec:proj} treats the computation of these weights in
detail, including the packet-sprayed case, where the weights are no
longer plain shortest-path distances. The projected graph contains \(n\)
vertices and \(O(n^2)\) edges, so, being dense, its minimum spanning tree
can be computed by Prim's algorithm in \(O(n^2)\) time. Hence one oracle
invocation requires
\begin{equation*}
T_{\mathrm{oracle}}
=
\ComplexityOracle .
\end{equation*}

DualColl is an instance of the Garg--K\"onemann fractional-packing scheme
applied to Equations~\eqref{eq:fractional-packing-objective}--%
\eqref{eq:fractional-packing-nonnegative}.
Theorem~\ref{thm:oracle-correctness} establishes that each iteration
selects an exact minimum-cost column, so DualColl inherits that scheme's
approximation guarantee and iteration bound directly.

\begin{theorem}[DualColl guarantee]
\label{thm:dualcoll-approximation}
For any \(0<\epsilon<1\), Algorithm~\ref{alg:dualcoll-framework} returns a
feasible fractional tree packing \(\widehat{\lambda}\) satisfying
\[
\sum_{T\in\mathcal{S}}\widehat{\lambda}_T
\;\geq\;
(1-\epsilon)^2\,\mathrm{OPT}.
\]
\end{theorem}

\begin{proof}[Proof of Theorem~\ref{thm:dualcoll-approximation}]
\label{proof:dualcoll-approximation}
Each iteration invokes the minimum-cost collective-tree oracle, which by
Theorem~\ref{thm:oracle-correctness} returns a tree minimizing
\(\sum_e A_{e,T}y_e\) exactly, i.e., an exact minimum-cost column of the
packing LP. Algorithm~\ref{alg:dualcoll-framework} uses this oracle within
the Garg--K\"onemann packing scheme with the standard
initialization \(y_e=\gamma/C_e\), bottleneck augmentation, multiplicative
update, dual-length stopping rule \(D(y)\ge1\), and final scaling by
\(\log_{1+\epsilon}\bigl((1+\epsilon)/\gamma\bigr)\). Under an exact
minimum-cost oracle, that scheme produces a feasible packing within a
\((1-\epsilon)^2\) factor of the optimum~\cite{garg2007faster}, which yields
the stated bound.
\end{proof}

Since every iteration performs one oracle call, and the Garg--K\"onemann
packing scheme terminates in
\[N_{\mathrm{iter}} = |E|\lceil (1/\epsilon)\log_{1+\epsilon}|E|\rceil
= \ComplexityNiter\] iterations~\cite{garg2007faster}, the
total running time is
\begin{equation*}
T_{\mathrm{DualColl}}
=
\ComplexityTotal .
\end{equation*}

DualColl therefore computes a near-optimal solution to the ForestColl
fractional tree-packing objective without constructing the
exponentially large set of candidate routed trees.

\paragraph{Generalizing to other collective operations.}
The same machinery solves the other collective operations with only a
change of oracle. For AllGather the objective is a max-min one,
maximizing the slowest node's aggregate broadcast bandwidth, and a packing
variable is a set of $n$ trees, one rooted at each compute node. The
oracle correspondingly returns, instead of a single minimum spanning tree,
a minimum spanning arborescence rooted at each node, computed on the same
projected graph from the same dual prices; everything outside the
oracle---the price updates, the augmentation rule, the stopping condition,
and the scaling---is unchanged. The resulting schedule is computed in time
$\ComplexityAllgather$, which improves on the bound reported for
ForestColl, as shown in Appendix~\ref{app:comparison}. ReduceScatter is obtained by the same change of
oracle, replacing each out-arborescence by an in-arborescence so that the
data flows toward the root rather than away from it. Multiple concurrent AllReduce sets are handled
by letting a packing variable be a $k$-tuple of trees, one spanning each
set, with the oracle computing the $k$ minimum-cost trees independently
and returning their union as a single column; the shared capacity
constraints then couple the sets automatically. We give the details in
Section~\ref{sec:other-collectives}.

The extensions to the routing and bandwidth-allocation models considered
in practical deployments are presented in
Section~\ref{sec:packet-spraying-mmf}.

\section{Packet Spraying and Max-Min Fairness}
\label{sec:packet-spraying-mmf}

The fractional packing algorithm of Section~\ref{sec:dualcoll-algorithm} computes an explicit bandwidth allocation under the scheduler-routed model. Namely, it returns a set of logical trees $T_1,\ldots,T_k$ and assigns each tree a continuous bandwidth $a_i$, with each logical edge mapped to a single chosen physical path. In that model, the scheduler directly controls both the route of every flow and the bandwidth of every tree, subject to the physical capacity constraints.

Real deployments differ from this idealization in two respects. First, when packet spraying is used, the fabric spreads a logical flow across several physical paths rather than allowing the scheduler to choose a path. Second, systems usually do not provide explicit per-tree bandwidth control; instead, bandwidth is allocated dynamically by the transport protocol and by the network, often according to a max-min fairness principle. This section addresses both. We first describe the routing mechanics of the packet-spraying model, which are independent of how bandwidth is subsequently allocated. We then define the tree-coupled version of max-min fairness appropriate for collective communication, generalize it to packet-sprayed flows, and finally show that by duplicating trees, the protocol model can approximate any explicit fractional allocation arbitrarily well.

Packet spraying also forces a change in the formulation itself, which we
flag at the outset. Under the scheduler-routed model we were entitled to
assume that a tree reduces and broadcasts along the \emph{same} logical
tree: Lemma~\ref{lem:same-tree} shows that any schedule with distinct
reduce and broadcast trees can be converted, at no loss of throughput,
into one built from such \emph{simple} trees. That argument reverses a
tree's routing and charges the reversed phase the mirrored load on the
opposite directed edges, which is valid precisely because single-path
routing sends a reversed logical edge back along the same physical path, flipping the induced loads $p_e(F)$ on every physical link.
Local equal next-hop spraying does not have this property. What survives
is the weaker structure of Lemmas~\ref{lem:single-root}
and~\ref{lem:phase-trees}, which make no reference to how a logical edge
is routed---they use only the dependency structure among flows and the
fact that the schedule does not split a shard---and therefore hold
verbatim under packet spraying: each shard still has a single root, with a reduce in-tree and a broadcast out-tree
rooted at it, but the two trees need not coincide. Accordingly, in the
packet-sprayed setting a packing variable is a \emph{pair} consisting of
an in-tree and an out-tree sharing a root, and the oracle must select the
best such pair rather than a single tree. Section~\ref{sec:ps-pairs}
develops this formulation and its oracle.

\subsection{Packet-Spraying Routing}
\label{sec:ps}

We first describe packet spraying purely as a routing mechanism, independent of how bandwidth is later allocated. Modern datacenter networks often contain several shortest paths between a pair of endpoints. In packet spraying, traffic between a pair of endpoints is split across the shortest paths. Since collective communication transfers large data chunks, we represent the long-run traffic distribution induced by packet spraying as fractional physical-link loads.

We consider \emph{local equal next-hop} packet spraying. At each switch, a
packet is forwarded uniformly at random among all outgoing next hops that
continue along a shortest path to its destination. The eligible next-hop
set is determined by the routing topology and does not change in response
to instantaneous congestion or Layer-2 PAUSE state. This forwarding rule
induces fixed long-run physical-link fractions for each logical flow,
which can therefore be incorporated directly into the fractional
tree-packing model. We do not consider adaptive packet spraying in which
congestion or flow-control state dynamically changes next-hop
eligibility.

For a directed logical flow $F=(u,v)$, let $p_e(F)\in[0,1]$ denote the
long-run fraction of the packets of $F$ that traverse physical directed
edge $e$. If $F$ does not use $e$, then $p_e(F)=0$; if all of $F$ uses
$e$, then $p_e(F)=1$; under packet spraying, intermediate values are
possible.

To state the forwarding rule precisely, fix a destination $v$ and, for a
physical node $c$, let
\begin{equation*}
\Gamma_v(c)
\end{equation*}
denote the set of outgoing neighbors of $c$ that continue along a shortest
path from $c$ to $v$. If a fraction $\phi(c)$ of the flow $F=(u,v)$ reaches
node $c$, then every eligible outgoing edge $(c,z)$ with $z\in\Gamma_v(c)$
carries the fraction
\begin{equation}
p_{(c,z)}(F) = \frac{\phi(c)}{|\Gamma_v(c)|}.
\label{eq:ps_split}
\end{equation}
The splitting is thus uniform \emph{locally} at each switch. This does not
imply that all complete shortest paths from $u$ to $v$ are selected with
equal probability, nor that all physical links lying on shortest paths
carry equal fractions; both are, in general, consequences of the local
rule rather than assumptions.

\paragraph{Computing the induced fractions.}
For each directed logical flow $F=(u,v)$, the fractions $p_e(F)$ are
obtained by propagating flow mass through the shortest-path DAG toward
$v$. Initialize $\phi(u)=1$ and $\phi(c)=0$ for every other node $c$, and
process the nodes in a topological order of the shortest-path DAG from $u$
toward $v$. For every processed node $c$ and every eligible outgoing edge
$e=(c,z)$ with $z\in\Gamma_v(c)$, set
\[
p_e(F) = \frac{\phi(c)}{|\Gamma_v(c)|},
\qquad
\phi(z) \mathrel{+}= p_e(F).
\]
When several upstream branches reach the same node, their contributions
are accumulated through $\phi(z)$. This propagation exactly captures the
long-run traffic fractions induced by independent uniform next-hop choices
at the switches.

It is worth distinguishing the two roles that shortest paths play in the
two models. In the scheduler-routed model of
Section~\ref{sec:dualcoll-algorithm}, the scheduler itself selects a
minimum-dual-cost physical path for each logical edge, using the current
prices $y_e$; the prices determine the route. In the packet-spraying
model, by contrast, the physical routing protocol fixes the shortest-path
DAG according to its own routing metric, independently of the dual
prices. The prices $y_e$ do not determine which links are eligible to
carry a flow; they are used only to evaluate the cost of the fixed
packet-spraying fractions $p_e(F)$ induced by that routing structure.

\subsection{The Packet-Sprayed Formulation}
\label{sec:ps-pairs}

By Lemmas~\ref{lem:single-root} and~\ref{lem:phase-trees}, a bandwidth-optimal AllReduce
schedule under packet spraying decomposes into shards, each with a single
root $r$, a reduce in-tree $R$ directed toward $r$, and a broadcast
out-tree $D$ directed away from $r$, with $R$ and $D$ not necessarily
equal. They cannot be assumed equal because reversing a logical edge no
longer mirrors its load onto the opposite directed edges: the fractions
$p_e(F)$ of a flow $F=(u,v)$ are determined by the shortest-path DAG
toward $v$ and the branching degrees along it, while the DAG toward $u$
may branch differently, so in general
$p_{\bar e}\bigl((v,u)\bigr)\neq p_e\bigl((u,v)\bigr)$ and the
reversal step in the proof of Lemma~\ref{lem:same-tree} does not apply.
We therefore take a packing variable to be such a \emph{pair}.

Let $\mathcal{P}$ be the set of all pairs $P=(R,D)$ in which $R$ is a
spanning in-arborescence and $D$ a spanning out-arborescence over the
participating ranks, both rooted at a common rank, together with the
packet-sprayed routing of their logical edges. A pair running at rate
$\lambda_P$ loads directed physical edge $e$ by
\begin{equation*}
A_{e,P} \;=\; \sum_{F\in R} p_e(F) \;+\; \sum_{F\in D} p_e(F),
\end{equation*}
the reduce and broadcast phases contributing separately because they are
routed independently. The fractional packing problem is then
\begin{equation}
\max_{\lambda\ge 0}\ \sum_{P\in\mathcal{P}} \lambda_P
\quad\text{subject to}\quad
\sum_{P\in\mathcal{P}} A_{e,P}\lambda_P \le C_e \ \ \forall e\in E,
\label{eq:ps-pair-packing}
\end{equation}
which is a fractional packing LP of the same form as the tree LP of
Section~\ref{sec:model}, with tree pairs instead of trees as its columns.

\paragraph{The minimum-cost pair oracle.}
Given dual prices $y$, the oracle must return
$P^\star\in\arg\min_{P\in\mathcal{P}}\sum_e A_{e,P}y_e$. Because the two
phases of a pair are routed independently, the cost separates:
\begin{equation*}
\mathrm{cost}_y(R,D) \;=\; \sum_{(u,v)\in R} W_{u,v}
\;+\; \sum_{(u,v)\in D} W_{u,v},
\end{equation*}
where $W_{u,v}=\sum_e p_e(u,v)y_e$ is the \emph{directed} projected weight
of the logical flow $(u,v)$, computed as in Section~\ref{sec:proj}. Note
that the projected graph is now directed: unlike the scheduler-routed
case, where the two directions of a logical tree edge were summed into a
single undirected weight $W_{\{u,v\}}$, here the reduce and broadcast
trees use different logical edges and each direction must be priced
separately.

The separation makes the oracle easy. For a \emph{fixed} root $r$, the
cheapest reduce tree and the cheapest broadcast tree may be chosen
independently, so the minimum-cost pair rooted at $r$ is obtained by
computing a minimum-weight spanning in-arborescence rooted at $r$ and a
minimum-weight spanning out-arborescence rooted at $r$ on the directed
projected graph. Minimizing over the $n$ choices of root gives $P^\star$.

\begin{theorem}[Pair-oracle correctness]
\label{thm:pair-oracle}
Under local equal next-hop packet spraying and route separability, the
procedure above returns a minimum-cost pair.
\end{theorem}

\begin{proof}[Proof of Theorem~\ref{thm:pair-oracle}]
\label{proof:pair-oracle}

Fix a root $r$. By route separability the physical routing of each logical
edge may be chosen independently of the others, so the minimum dual cost
of a pair rooted at $r$ is the sum of its logical edges' projected
weights. The costs of the reduce in-tree and the broadcast out-tree add,
so the pair cost is minimized by minimizing
each separately; a minimum-weight spanning in-arborescence and a
minimum-weight spanning out-arborescence rooted at $r$ achieve these two
minima. Hence the procedure returns the cheapest pair rooted at $r$, and
taking the best over all $n$ roots returns the cheapest pair overall.

\end{proof}

Here Algorithm~\ref{alg:dualcoll-framework} is applied to the new packing
LP over tree pairs, Equation~\eqref{eq:ps-pair-packing}, rather than to the
tree LP of Section~\ref{sec:model}; only its columns and its oracle change.
Since the oracle is exact, the algorithm retains its guarantee: the packing
it returns is within $(1-\epsilon)^2$ of the
optimum of Equation~\eqref{eq:ps-pair-packing}, which by
Lemmas~\ref{lem:single-root} and~\ref{lem:phase-trees} is the bandwidth-optimal AllReduce
throughput under packet spraying.

\subsection{Efficient Computation of Projection Weights}
\label{sec:proj}

In the scheduler-routed model, the projection weight
$W_{\{u,v\}}=d_y(u,v)+d_y(v,u)$ of Section~\ref{sec:collective-tree-oracle}
is built from shortest-path distances under the dual weights $y_e$,
computed by running Dijkstra from every rank in
$\ComplexityDijkstraAllRanks$ total. This subsection treats the
packet-spraying (PS) model of Section~\ref{sec:ps}, where a logical flow
is instead spread fractionally across the shortest paths. For a logical
flow $F=(u,v)$, the projection weight is the split-weighted sum of dual
weights,
\begin{equation}
W_{u,v} = \sum_{e \in E} p_e(F)\, y_e ,
\label{eq:proj_weight_ps}
\end{equation}
where $p_e(F)$ is the long-run fraction of $F$ routed over link $e$.

Equation~\eqref{eq:proj_weight_ps} is the expected dual cost paid by a
packet of $F$, the expectation being over the independent uniform
next-hop choices made at the switches. This observation yields a direct
recurrence that avoids materializing the fractions $p_e(F)$ at all.

\paragraph{A destination-rooted dynamic program.}
Fix a destination $v$ and recall from Section~\ref{sec:ps} that
$\Gamma_v(c)$ denotes the set of outgoing neighbors of a physical node
$c$ that continue along a shortest path from $c$ to $v$. Define
$D_v(c)$ to be the expected dual cost incurred by a packet that is
currently at $c$ and is forwarded under local equal next-hop spraying
until it reaches $v$. Then $D_v(v)=0$ and, for every other node $c$,
\begin{equation}
D_v(c)
=
\frac{1}{|\Gamma_v(c)|}
\sum_{z\in\Gamma_v(c)}
\bigl( y_{(c,z)} + D_v(z) \bigr) .
\label{eq:ps-dp}
\end{equation}
Each eligible next hop is taken with probability $1/|\Gamma_v(c)|$, and
taking it costs $y_{(c,z)}$ immediately plus the expected cost $D_v(z)$
from $z$ onward, so Equation~\eqref{eq:ps-dp} is exactly the expected-cost
recurrence of the forwarding rule. Consequently, for every source $u$,
\begin{equation*}
W_{u,v} = D_v(u) .
\end{equation*}
The recurrence is well founded: every $z\in\Gamma_v(c)$ is strictly closer
to $v$ than $c$ is, so evaluating the nodes in increasing distance to $v$---equivalently, in reverse topological order of the shortest-path DAG
into $v$---computes every $D_v(c)$ from values already known. When
shortest paths are defined by hop count, that ordering is produced by a
single backward breadth-first search from $v$.

\paragraph{Cost.}
For a fixed destination $v$, one traversal builds the shortest-path DAG
into $v$ and evaluates Equation~\eqref{eq:ps-dp} at every node, touching
each physical edge a constant number of times, so it costs
$\ComplexityProjectionDPSingle$. Running it once per destination rank
yields \emph{all} $n^2$ directed projection weights $W_{u,v}$, since the
sweep for $v$ produces $D_v(u)$ for every $u$ simultaneously. The total
projection cost per oracle iteration is therefore
\begin{equation}
T_{\mathrm{proj}}^{\mathrm{PS}} = \ComplexityProjectionDPAll .
\label{eq:ps-projection-cost}
\end{equation}
Unlike a precomputation of the fractions $p_e(F)$, this must be repeated
each iteration, because the dual weights $y_e$ change between iterations;
the shortest-path DAGs themselves, however, depend only on the topology
and may be built once and reused.

\subsection{Max-Min Fair Realization}

We now turn to how bandwidth is allocated once trees are instantiated. In the packet-sprayed formulation, we treat each reduce/broadcast tree pair as one collective object whose constituent flows share a common rate; all definitions below apply to such pairs by summing the loads of both phases. We first define tree-coupled max-min fairness for single-path routing, and then generalize it to packet-sprayed flows.

In ordinary flow-level max-min fairness, each flow is treated independently. A collective tree is different: it is not a single physical flow, but a collection of physical flows corresponding to the logical edges of the tree. Moreover, the tree can only make progress at the rate of its slowest flow. If one edge of the tree is bottlenecked, the whole tree pipeline slows down to that rate.

We therefore use a \emph{tree-coupled max-min fairness} model. Each tree $T_i$ consists of a set of physical flows induced by routing its logical edges over the physical topology. The bandwidth $B(T_i)$ of the tree is the common rate at which all of its flows operate. If any flow of $T_i$ becomes bottlenecked, then all other flows of $T_i$ are clamped to the same rate. This captures the synchronous nature of a tree-based collective: a tree cannot complete its assigned data chunk faster than its slowest required communication path.

Assume first that every logical edge is routed over a unique physical path. The tree-coupled max-min allocation can be described by the following iterative procedure.

It is important to note that this procedure describes \emph{how the resulting allocation is calculated}, not a computation that the network carries out. The bandwidth given to each tree is determined by the network itself, dynamically, according to its transport and congestion-control protocol. The network does not run this algorithm; rather, the algorithm is our model for predicting the rates that the network's max-min fair sharing will produce, so that the scheduler can reason about them when choosing which trees to instantiate. Tree-coupled max-min fairness is thus an idealized analytical model of protocol-determined bandwidth allocation. We do not claim that any particular congestion-control protocol realizes exact max-min fairness; rather, max-min fairness is a standard and analytically tractable approximation of the rates that such protocols tend to produce, and we use it to reason about the effect of tree duplication.

\paragraph{Bottleneck selection.}
For every physical link $e$, compute the ratio between its remaining capacity and the number of active tree flows that traverse it. Let
\begin{equation}
B^* =
\min_{e \in E}
\frac{C_e^{\mathrm{rem}}}{n_e},
\label{eq:tree_coupled_bottleneck}
\end{equation}
where $C_e^{\mathrm{rem}}$ is the remaining capacity on link $e$, initially equal to $C_e$, and $n_e$ is the number of active tree flows using $e$.

\paragraph{Tree clamping.}
Let $e^*$ be a link attaining the minimum in Equation~\eqref{eq:tree_coupled_bottleneck}. Every active tree that uses $e^*$ is assigned bandwidth $B^*$. Once a tree is assigned this bandwidth, all of its flows, on all physical links, are clamped to rate $B^*$.

\paragraph{Capacity update.}
For every newly clamped tree, subtract its bandwidth consumption from every physical link used by that tree. Remove the clamped trees from the active set, and repeat the procedure until all trees have been assigned bandwidths.

This procedure differs from ordinary flow-level max-min fairness only in the clamping step. In ordinary fairness, a bottlenecked flow is fixed but unrelated flows remain independent. In tree-coupled fairness, once one flow of a tree is bottlenecked, the entire tree is fixed to the same rate.

A basic property of this procedure, which we use later when analyzing tree duplication, is that the rate at which trees are clamped never decreases as the algorithm proceeds. We state and prove this monotonicity, in a general form that also covers packet spraying, in Lemma~\ref{lem:monotone_clamp} at the end of this subsection.

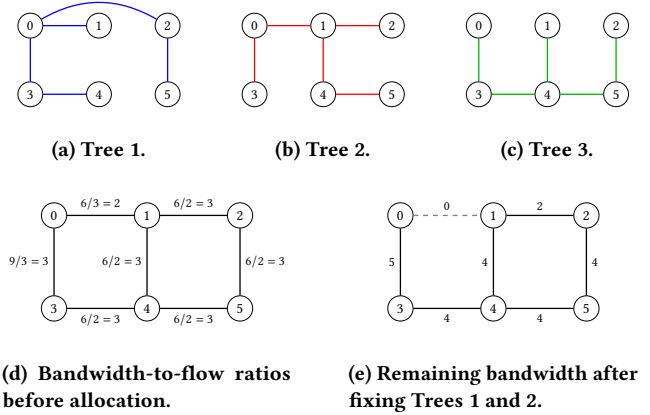
\begin{figure}[tbp]
\centering
\begin{subfigure}[t]{0.30\linewidth}
  \centering
  \resizebox{\linewidth}{!}{%
  \begin{tikzpicture}
    \useasboundingbox (-0.6,-2.1) rectangle (3.6,0.6);
    \node[circle, draw] (0) at (0, 0) {0};
    \node[circle, draw] (1) at (1.5, 0) {1};
    \node[circle, draw] (2) at (3, 0) {2};
    \node[circle, draw] (3) at (0, -1.5) {3};
    \node[circle, draw] (4) at (1.5, -1.5) {4};
    \node[circle, draw] (5) at (3, -1.5) {5};
    \draw[thick, blue] (0) -- (1);
    \draw[thick, blue, bend left=30] (0) to (2);
    \draw[thick, blue] (0) -- (3);
    \draw[thick, blue] (3) -- (4);
    \draw[thick, blue] (2) -- (5);
  \end{tikzpicture}%
  }
  \caption{Tree 1.}
  \label{fig:maxmin-t1}
\end{subfigure}
\hfill
\begin{subfigure}[t]{0.30\linewidth}
  \centering
  \resizebox{\linewidth}{!}{%
  \begin{tikzpicture}
    \useasboundingbox (-0.6,-2.1) rectangle (3.6,0.6);
    \node[circle, draw] (0) at (0, 0) {0};
    \node[circle, draw] (1) at (1.5, 0) {1};
    \node[circle, draw] (2) at (3, 0) {2};
    \node[circle, draw] (3) at (0, -1.5) {3};
    \node[circle, draw] (4) at (1.5, -1.5) {4};
    \node[circle, draw] (5) at (3, -1.5) {5};
    \draw[thick, red] (0) -- (1);
    \draw[thick, red] (1) -- (2);
    \draw[thick, red] (0) -- (3);
    \draw[thick, red] (1) -- (4);
    \draw[thick, red] (4) -- (5);
  \end{tikzpicture}%
  }
  \caption{Tree 2.}
  \label{fig:maxmin-t2}
\end{subfigure}
\hfill
\begin{subfigure}[t]{0.30\linewidth}
  \centering
  \resizebox{\linewidth}{!}{%
  \begin{tikzpicture}
    \useasboundingbox (-0.6,-2.1) rectangle (3.6,0.6);
    \node[circle, draw] (0) at (0, 0) {0};
    \node[circle, draw] (1) at (1.5, 0) {1};
    \node[circle, draw] (2) at (3, 0) {2};
    \node[circle, draw] (3) at (0, -1.5) {3};
    \node[circle, draw] (4) at (1.5, -1.5) {4};
    \node[circle, draw] (5) at (3, -1.5) {5};
    \draw[thick, green!70!black] (0) -- (3);
    \draw[thick, green!70!black] (3) -- (4);
    \draw[thick, green!70!black] (4) -- (5);
    \draw[thick, green!70!black] (5) -- (2);
    \draw[thick, green!70!black] (1) -- (4);
  \end{tikzpicture}%
  }
  \caption{Tree 3.}
  \label{fig:maxmin-t3}
\end{subfigure}

\vspace{0.15cm}

\begin{subfigure}[t]{0.45\linewidth}
  \centering
  \resizebox{\linewidth}{!}{%
  \begin{tikzpicture}
    \useasboundingbox (-1.1,-2.8) rectangle (5.1,0.8);
    \node[circle, draw] (0) at (0, 0) {0};
    \node[circle, draw] (1) at (2, 0) {1};
    \node[circle, draw] (2) at (4, 0) {2};
    \node[circle, draw] (3) at (0, -2) {3};
    \node[circle, draw] (4) at (2, -2) {4};
    \node[circle, draw] (5) at (4, -2) {5};

    \draw[thick] (0) -- node[above] {\small $6/3 = 2$} (1);
    \draw[thick] (1) -- node[above] {\small $6/2 = 3$} (2);
    \draw[thick] (0) -- node[left] {\small $9/3 = 3$} (3);
    \draw[thick] (1) -- node[left] {\small $6/2 = 3$} (4);
    \draw[thick] (2) -- node[right] {\small $6/2 = 3$} (5);
    \draw[thick] (3) -- node[below] {\small $6/2 = 3$} (4);
    \draw[thick] (4) -- node[below] {\small $6/2 = 3$} (5);
  \end{tikzpicture}%
  }
  \caption{Bandwidth-to-flow ratios before allocation.}
  \label{fig:maxmin-initial}
\end{subfigure}
\hfill
\begin{subfigure}[t]{0.45\linewidth}
  \centering
  \resizebox{\linewidth}{!}{%
  \begin{tikzpicture}
    \useasboundingbox (-1.1,-2.8) rectangle (5.1,0.8);
    \node[circle, draw] (0) at (0, 0) {0};
    \node[circle, draw] (1) at (2, 0) {1};
    \node[circle, draw] (2) at (4, 0) {2};
    \node[circle, draw] (3) at (0, -2) {3};
    \node[circle, draw] (4) at (2, -2) {4};
    \node[circle, draw] (5) at (4, -2) {5};

    \draw[thick, dashed, gray] (0) -- node[above, text=black] {\small 0} (1);
    \draw[thick] (1) -- node[above] {\small 2} (2);
    \draw[thick] (0) -- node[left] {\small 5} (3);
    \draw[thick] (1) -- node[left] {\small 4} (4);
    \draw[thick] (2) -- node[right] {\small 4} (5);
    \draw[thick] (3) -- node[below] {\small 4} (4);
    \draw[thick] (4) -- node[below] {\small 4} (5);
  \end{tikzpicture}%
  }
  \caption{Remaining bandwidth after fixing Trees 1 and 2.}
  \label{fig:maxmin-remaining}
\end{subfigure}
\caption{\textbf{Tree-coupled max-min fairness on a $2 \times 3$ partial-mesh topology.}
Three logical trees, shown in (a), (b), and (c), are mapped to a physical topology in which all physical links have capacity 6, except link $(0,3)$, which has capacity 9.
In (d), each physical link is labeled by the ratio between its capacity and the number of active tree flows that traverse it. The minimum ratio is 2, attained by link $(0,1)$, and therefore Trees 1 and 2 are bottlenecked at bandwidth 2.
In (e), the remaining capacity is shown after fixing Trees 1 and 2. Tree 3 is then assigned bandwidth 4.}
\label{fig:maxmin}
\Description{Five diagrams of a two-by-three mesh of six nodes. Panels (a),
(b) and (c) each highlight the edges of one logical tree. Panel (d) labels
every link with its capacity divided by the number of tree flows crossing
it; link (0,1) has the smallest such ratio, 2. Panel (e) shows the
capacities remaining after the two trees that use link (0,1) are fixed at
bandwidth 2.}
\end{figure}

Figure~\ref{fig:maxmin} illustrates the procedure on a $2 \times 3$ partial mesh with six nodes, in which all physical links have capacity 6, except link $(0,3)$, which has capacity 9. The schedule contains three logical trees. After mapping logical edges to physical shortest paths, we compute for each physical link the ratio between its capacity and the number of active tree flows that traverse it. The minimum ratio is 2, attained by link $(0,1)$. Therefore, all trees that use this link are bottlenecked; in the example, these are Trees 1 and 2. Both are assigned bandwidth 2, and all of their flows are clamped to rate 2. After subtracting their bandwidth consumption from the network, Tree 3 remains. Its minimum residual capacity is 4, and hence it is assigned bandwidth 4.

\paragraph{Generalization to packet spraying.}
The bottleneck computation above assumes each logical edge is routed over a single physical path. Under the packet-spraying model of Section~\ref{sec:ps}, a logical flow $F$ instead loads each physical link $e$ fractionally by $p_e(F)$. The bottleneck computation from Equation~\eqref{eq:tree_coupled_bottleneck} generalizes by replacing the number of active flows on a link with the total fractional load induced by active flows:
\begin{equation}
B^* =
\min_{e \in E}
\frac{C_e^{\mathrm{rem}}}
{\sum_{F \in \mathcal{F}_{\mathrm{active}}} p_e(F)} .
\label{eq:ecmp_bottleneck}
\end{equation}
Here, $\mathcal{F}_{\mathrm{active}}$ is the set of constituent logical flows belonging to trees that have not yet been clamped.

Suppose link $e^*$ attains the minimum in Equation~\eqref{eq:ecmp_bottleneck}. Any active tree with a constituent flow $F$ such that $p_{e^*}(F)>0$ is bottlenecked at bandwidth $B^*$. The amount of physical bandwidth consumed by this flow on link $e^*$ is then $p_{e^*}(F)B^*$, and the total rate of the parent logical flow is $B^*$. Thus, the same tree-coupled clamping rule established for single-path routing applies: once a tree is bottlenecked on any fraction of any constituent flow, the entire tree is assigned the bottleneck rate $B^*$. In particular, all trees sharing the bottlenecked link $e^*$ are clamped to the \emph{same} tree-level rate $B^*$, even though their individual sub-flows on $e^*$ may consume different physical amounts $p_{e^*}(F)B^*$; the equal-rate property holds at the level of the original trees, as shown for the single-path case.

The full procedure, including both the single-path and packet-spraying cases, is summarized in Algorithm~\ref{alg:tree-coupled-maxmin}.

\begin{algorithm}[t]
\caption{Tree-coupled max-min fair allocation under packet spraying}
\label{alg:tree-coupled-maxmin}
\begin{algorithmic}[1]
\Require Trees $\{T_1,\ldots,T_k\}$ with constituent flows and split
         coefficients $p_e(F)$ for every flow $F$ and link $e$; link
         capacities $\{C_e\}_{e\in E}$
\Ensure Per-tree bandwidths $\{B(T_i)\}_{i=1}^k$
\State $\mathcal{A} \gets \{T_1,\ldots,T_k\}$
       \Comment{active (unclamped) trees}
\State $C_e^{\mathrm{rem}} \gets C_e$ for every $e \in E$
       \Comment{residual capacities}
\While{$\mathcal{A} \neq \emptyset$}
    \State $\mathcal{F}_{\mathrm{active}} \gets
           \{\,F : F \text{ is a flow of some } T \in \mathcal{A}\,\}$
    \ForAll{$e \in E$ with $\sum_{F\in\mathcal{F}_{\mathrm{active}}} p_e(F) > 0$}
        \State $r_e \gets
               \dfrac{C_e^{\mathrm{rem}}}
                     {\sum_{F\in\mathcal{F}_{\mathrm{active}}} p_e(F)}$
               \Comment{fair-share ratio of link $e$}
    \EndFor
    \State $B^\star \gets \min_e r_e$

    \State $e^\star \gets \arg\min_e r_e$
           \Comment{bottleneck link}
    \State $\mathcal{C} \gets
           \{\,T \in \mathcal{A} : \exists\, \text{flow } F \text{ of } T
           \text{ with } p_{e^\star}(F) > 0\,\}$

    \ForAll{$T \in \mathcal{C}$}
        \State $B(T) \gets B^\star$
               \Comment{clamp the whole tree}
    \EndFor
    \ForAll{$e \in E$}
        \State $C_e^{\mathrm{rem}} \gets
               C_e^{\mathrm{rem}} - B^\star
               \sum_{F \text{ a flow of some } T \in \mathcal{C}}
               p_e(F)$
     
    \EndFor
    \State $\mathcal{A} \gets \mathcal{A} \setminus \mathcal{C}$
           \Comment{deactivate clamped trees}
\EndWhile
\State \Return $\{B(T_i)\}_{i=1}^k$
\end{algorithmic}
\end{algorithm}

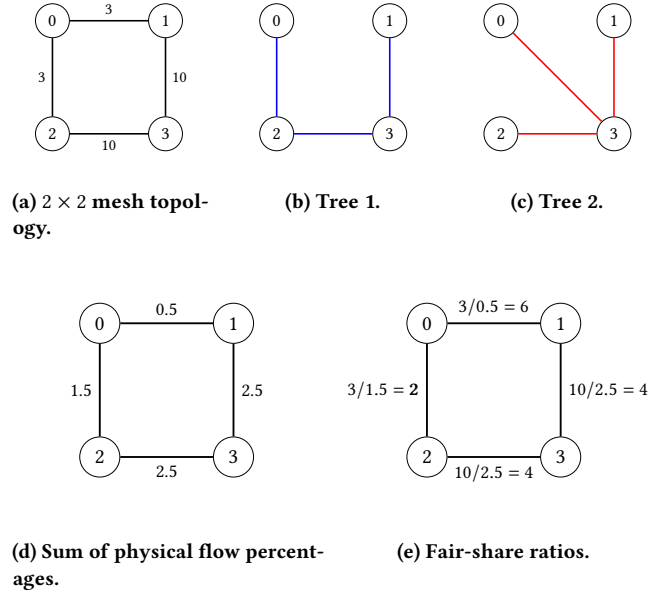
\begin{figure}[htbp]
\centering
\begin{subfigure}[t]{0.30\linewidth}
  \centering
  \resizebox{\linewidth}{!}{%
  \begin{tikzpicture}
    \useasboundingbox (-0.7,-2.7) rectangle (2.7,0.7);
    \node[circle, draw, minimum size=0.6cm] (0) at (0, 0) {0};
    \node[circle, draw, minimum size=0.6cm] (1) at (2, 0) {1};
    \node[circle, draw, minimum size=0.6cm] (2) at (0, -2) {2};
    \node[circle, draw, minimum size=0.6cm] (3) at (2, -2) {3};
    \draw[thick] (0) -- node[above] {\small 3} (1);
    \draw[thick] (0) -- node[left] {\small 3} (2);
    \draw[thick] (1) -- node[right] {\small 10} (3);
    \draw[thick] (2) -- node[below] {\small 10} (3);
  \end{tikzpicture}%
  }
  \caption{$2 \times 2$ mesh topology.}
  \label{fig:ecmp-topology}
\end{subfigure}
\hfill
\begin{subfigure}[t]{0.30\linewidth}
  \centering
  \resizebox{\linewidth}{!}{%
  \begin{tikzpicture}
    \useasboundingbox (-0.7,-2.7) rectangle (2.7,0.7);
    \node[circle, draw, minimum size=0.6cm] (0) at (0, 0) {0};
    \node[circle, draw, minimum size=0.6cm] (1) at (2, 0) {1};
    \node[circle, draw, minimum size=0.6cm] (2) at (0, -2) {2};
    \node[circle, draw, minimum size=0.6cm] (3) at (2, -2) {3};
    \draw[thick, blue] (0) -- (2);
    \draw[thick, blue] (2) -- (3);
    \draw[thick, blue] (3) -- (1);
  \end{tikzpicture}%
  }
  \caption{Tree 1.}
  \label{fig:ecmp-tree1}
\end{subfigure}
\hfill
\begin{subfigure}[t]{0.30\linewidth}
  \centering
  \resizebox{\linewidth}{!}{%
  \begin{tikzpicture}
    \useasboundingbox (-0.7,-2.7) rectangle (2.7,0.7);
    \node[circle, draw, minimum size=0.6cm] (0) at (0, 0) {0};
    \node[circle, draw, minimum size=0.6cm] (1) at (2, 0) {1};
    \node[circle, draw, minimum size=0.6cm] (2) at (0, -2) {2};
    \node[circle, draw, minimum size=0.6cm] (3) at (2, -2) {3};
    \draw[thick, red] (0) -- (3);
    \draw[thick, red] (3) -- (2);
    \draw[thick, red] (3) -- (1);
  \end{tikzpicture}%
  }
  \caption{Tree 2.}
  \label{fig:ecmp-tree2}
\end{subfigure}

\vspace{0.15cm}

\begin{subfigure}[t]{0.48\linewidth}
  \centering
  \resizebox{\linewidth}{!}{%
  \begin{tikzpicture}
    \useasboundingbox (-1.3,-3.0) rectangle (3.3,1.0);
    \node[circle, draw, minimum size=0.6cm] (0) at (0, 0) {0};
    \node[circle, draw, minimum size=0.6cm] (1) at (2, 0) {1};
    \node[circle, draw, minimum size=0.6cm] (2) at (0, -2) {2};
    \node[circle, draw, minimum size=0.6cm] (3) at (2, -2) {3};
    \draw[thick] (0) -- node[above] {\small 0.5} (1);
    \draw[thick] (0) -- node[left] {\small 1.5} (2);
    \draw[thick] (1) -- node[right] {\small 2.5} (3);
    \draw[thick] (2) -- node[below] {\small 2.5} (3);
  \end{tikzpicture}%
  }
  \caption{Sum of physical flow percentages.}
  \label{fig:ecmp-flow-percentages}
\end{subfigure}
\hfill
\begin{subfigure}[t]{0.48\linewidth}
  \centering
  \resizebox{\linewidth}{!}{%
  \begin{tikzpicture}
    \useasboundingbox (-1.3,-3.0) rectangle (3.3,1.0);
    \node[circle, draw, minimum size=0.6cm] (0) at (0, 0) {0};
    \node[circle, draw, minimum size=0.6cm] (1) at (2, 0) {1};
    \node[circle, draw, minimum size=0.6cm] (2) at (0, -2) {2};
    \node[circle, draw, minimum size=0.6cm] (3) at (2, -2) {3};
    \draw[thick] (0) -- node[above] {\small $3 / 0.5 = 6$} (1);
    \draw[thick] (0) -- node[left] {\small $3 / 1.5 = \textbf{2}$} (2);
    \draw[thick] (1) -- node[right] {\small $10 / 2.5 = 4$} (3);
    \draw[thick] (2) -- node[below] {\small $10 / 2.5 = 4$} (3);
  \end{tikzpicture}%
  }
  \caption{Fair-share ratios.}
  \label{fig:ecmp-fair-share}
\end{subfigure}
\caption{\textbf{Tree-coupled max-min fairness under packet spraying.}
(a) A $2 \times 2$ physical mesh with non-uniform link capacities.
(b) Tree 1 is routed using explicit physical paths.
(c) Tree 2 contains the logical edge $(0,3)$, which is split by packet spraying across the two shortest physical paths from node 0 to node 3.
(d) The resulting sum of physical flow fractions after accounting for the packet-spraying split.
(e) The fair-share ratio, computed as capacity divided by accumulated flow fraction, identifies link $(0,2)$ as the global bottleneck.}
\label{fig:ecmp}
\Description{Five diagrams of a two-by-two mesh of four nodes with unequal
link capacities. Panel (a) shows the topology, panels (b) and (c) the edges
of two logical trees, where the second tree's logical edge from node 0 to
node 3 is split across the two shortest paths. Panel (d) labels each link
with the summed fractional flow it carries, and panel (e) labels each link
with capacity divided by that fraction, identifying link (0,2) as the
bottleneck.}
\end{figure}

Figure~\ref{fig:ecmp} gives an example on a $2 \times 2$ mesh. Tree 2 contains a logical edge from node 0 to node 3. Since there is no direct physical link between them, packet spraying splits this logical flow equally between the two shortest paths, $0$--$1$--$3$ and $0$--$2$--$3$. Therefore, this logical flow contributes $1/2$ unit of load to each physical link on these two paths. After summing the fractional contributions of all active tree flows, the minimum fair-share ratio is attained by link $(0,2)$ and equals 2. Since both trees use this link, either fully or fractionally, both trees are clamped to bandwidth 2.

\paragraph{Correctness under fractional loads.}
The clamping procedure above is an instance of the water-filling
algorithm: it repeatedly raises a common rate on all unclamped trees until
a link saturates, then freezes the trees on that link and continues. The
following theorem states that this procedure is correct even when the
per-link loads $p_e(F)$ are fractional, so that the rates it computes are
exactly the tree-coupled max-min fair ones.

\begin{theorem}[The computed allocation is max-min fair]
\label{thm:maxmin-correct}
For arbitrary non-negative per-link loads $p_e(F)$, the tree bandwidths
$\{B(T_i)\}$ produced by Algorithm~\ref{alg:tree-coupled-maxmin} are the
tree-coupled max-min fair allocation over the feasible set of tree-rate
vectors.
\end{theorem}

The proof rests on a structural property of the feasible set. Water-filling computes the max-min fair allocation on any feasible set
with the \emph{free-disposal} property~\cite[Theorems~2 and~3]{radunovic2007}.
Writing $\vec e_i$ for the unit vector along coordinate $i$, a set
$\mathcal{X} \subseteq \mathbb{R}^N$ has this property if there is a
lower-bound vector $\vec m$ with $\vec x \ge \vec m$ for all
$\vec x \in \mathcal{X}$, and for every $\vec x \in \mathcal{X}$, every
coordinate $i$, and every $\alpha \ge 0$ with
$\vec x - \alpha\, \vec e_i \ge \vec m$, the decreased vector
$\vec x - \alpha\, \vec e_i$ is also in $\mathcal{X}$. In words, lowering a
single coordinate, as long as it stays above its lower bound, never leaves
the set.

\begin{lemma}[The tree polytope is free-disposal]
\label{lem:free-disposal}
Let $T_1,\ldots,T_N$ be trees, and let $q_{e,i}\ge0$ be the load that tree
$T_i$ places on directed edge $e$ per unit of rate. Then the polytope of
feasible tree-rate vectors
\begin{equation*}
\mathcal{B} =
\Bigl\{\, B \ge 0 \ : \
\textstyle\sum_{i} q_{e,i}\, B_i \le C_e \ \ \forall e \in E \,\Bigr\},
\end{equation*}
in which each coordinate $B_i$ is the rate of tree $T_i$, has the
free-disposal property.
\end{lemma}

\begin{proof}
Each coordinate $B_i$ is bounded below by $0$, so the required lower bound holds with $m = 0$. Let $B \in \mathcal{B}$, fix a tree index $i_0$, and let $B' = B - \alpha\, \vec e_{i_0}$ for some $\alpha \ge 0$ with $B' \ge 0$; that is, $B'$ lowers the rate of the single tree $T_{i_0}$ while keeping it non-negative. For every link $e$,
\begin{equation*}
\sum_{i} q_{e,i}\, B'_i
=
\sum_{i} q_{e,i}\, B_i \ -\ \alpha\, q_{e,i_0}
\ \le\
\sum_{i} q_{e,i}\, B_i
\ \le\
C_e,
\end{equation*}
where the first inequality uses $\alpha \ge 0$ and $q_{e,i_0} \ge 0$. Hence every capacity constraint remains satisfied and $B' \in \mathcal{B}$. Since this holds for any tree coordinate and any admissible $\alpha$, the set $\mathcal{B}$ has the free-disposal property. The argument uses only $q_{e,i} \ge 0$, and is therefore unaffected by whether the underlying loads are integer or fractional.
\end{proof}

\begin{proof}[Proof of Theorem~\ref{thm:maxmin-correct}]
\label{proof:maxmin-correct}
We express the capacity constraints in tree-rate coordinates. Define the
\emph{aggregate load coefficient} of tree $T_i$ on link $e$ as the total
load its flows place on $e$ per unit of tree rate,
\begin{equation*}
q_{e,i} = \sum_{F \in T_i} p_e(F) \ \ge 0 ,
\end{equation*}
so that a tree running at rate $B_i$ consumes $q_{e,i} B_i$ on link $e$
(for a reduce/broadcast tree pair, the sum runs over the flows of both
phases),
and the feasible set of tree-rate vectors is the polytope $\mathcal{B}$ of
Lemma~\ref{lem:free-disposal}. By that lemma $\mathcal{B}$ is
free-disposal, and water-filling computes the max-min fair allocation on
any free-disposal set~\cite[Theorems~2 and~3]{radunovic2007}. Since
Algorithm~\ref{alg:tree-coupled-maxmin} is water-filling on
$\mathcal{B}$, the rates it computes are the max-min fair allocation.
\end{proof}

We also record two properties of the tree-coupled max-min procedure. The
first is the monotonicity announced above; the second is used in
Section~\ref{sec:duplication} to analyze tree duplication.

\begin{lemma}[Monotonicity of the clamp rate]
\label{lem:monotone_clamp}
Let $B^*_1, B^*_2, \ldots$ be the sequence of bottleneck values $B^*$ computed in successive iterations of the tree-coupled max-min procedure. Then this sequence is non-decreasing: $B^*_{t} \le B^*_{t+1}$ for every iteration $t$.
\end{lemma}

\begin{proof}
Under packet spraying, a link $e$ carries a fractional load $\ell_e = \sum_{F} p_e(F)$, summed over the flows of active trees, and its fair-share ratio is $C^{\mathrm{rem}}_e / \ell_e$. Single-path routing is the special case in which every $p_e(F) \in \{0,1\}$, so that $\ell_e = n_e$ is an integer; it therefore suffices to prove the claim for the fractional case. We caution that the rate at which individual sub-flows are clamped need not be monotonic, since removing a clamped tree subtracts fractional loads that differ across links; the quantity shown to be monotonic is the tree-level bottleneck value $B^*$, taken as the global minimum over all links in each iteration.

Consider two consecutive iterations $t$ and $t+1$, with bottleneck value $B^*_t$ attained at some link at iteration $t$. Let $e$ be any link that still has active flows at iteration $t+1$. Between the two iterations, the capacity-update step removed the trees clamped at iteration $t$, each of which had been assigned rate exactly $B^*_t$. Let $\delta_e \ge 0$ be the total fractional load these clamped trees contributed on link $e$. Removing them subtracts $\delta_e B^*_t$ from the remaining capacity of $e$ (the physical bandwidth they consumed) and $\delta_e$ from its fractional load. Writing the iteration-$t$ ratio of $e$ as $C/\ell$, the iteration-$(t+1)$ ratio is
\begin{equation*}
\frac{C - \delta_e B^*_t}{\ell - \delta_e}.
\end{equation*}
Since $B^*_t$ is the global minimum fair-share ratio over all links at iteration $t$, every link---including $e$---satisfies $C/\ell \ge B^*_t$ at that iteration, i.e., $C \ge \ell B^*_t$. Substituting,
\begin{equation*}
\frac{C - \delta_e B^*_t}{\ell - \delta_e}
\ge
\frac{\ell B^*_t - \delta_e B^*_t}{\ell - \delta_e}
=
B^*_t .
\end{equation*}
Thus every link that still has active flows on it at iteration $t+1$ has fair-share ratio at least $B^*_t$. Since $B^*_{t+1}$ is the minimum of these ratios, $B^*_{t+1} \ge B^*_t$.
\end{proof}

\begin{lemma}[Max-min lower bound]
\label{lem:maxmin_lower_bound}
Consider a set of active tree copies under tree-coupled max-min fairness. If there exists a feasible allocation in which every copy receives rate at least $X$, then the tree-coupled max-min fair allocation assigns every copy rate at least $X$.
\end{lemma}

\begin{proof}
Assume for contradiction that some tree copy receives a rate below $X$ under tree-coupled max-min fairness. Consider the first iteration of the tree-coupled procedure at which a copy is clamped to a rate below $X$, and let this rate be $B^* < X$, attained at bottleneck link $e$. Because this is the \emph{first} such iteration, every copy clamped at an earlier iteration received a rate at least $X > B^*$. By the monotonicity of the clamp rate (Lemma~\ref{lem:monotone_clamp}), the clamp rates are non-decreasing, so no earlier copy could have been clamped at a rate exceeding $B^*$; combined with the previous sentence, this means this is the first iteration of the algorithm, and every copy using $e$ is still active and about to be clamped at $B^*$.

At this iteration, the fair-share ratio of $e$ is
\begin{equation*}
B^* = \frac{C_e^{\mathrm{rem}}}{\sum_{F} p_e(F)} < X,
\end{equation*}
where the sum is over the constituent flows of the copies still active on $e$, and $C_e^{\mathrm{rem}} = C_e$ since no copy has yet been clamped and had its consumption removed. Rearranging, $\sum_F p_e(F)\, X > C_e$. But this states exactly that if every active copy on $e$ were assigned rate $X$, the total physical load on $e$ would exceed its capacity $C_e$. This contradicts the assumed existence of a feasible allocation in which every copy receives rate at least $X$, since that allocation would also violate the capacity of $e$. Hence no copy is clamped below $X$, and tree-coupled max-min fairness assigns every copy rate at least $X$.
\end{proof}

\subsection{Emulating Explicit Bandwidths}
\label{sec:duplication}

We now show how to emulate an explicit bandwidth allocation using only max-min fair sharing. The key idea is tree duplication. If one tree should receive twice as much bandwidth as another, we instantiate two identical copies of the first tree and one copy of the second. Under max-min fairness, identical competing copies receive equal rates; the total bandwidth assigned to the original tree is the sum of the rates of its copies.

For example, suppose the explicit allocation assigns bandwidth 100 to Tree A and bandwidth 50 to Tree B. If only these two trees are instantiated, a max-min fair protocol may not produce the desired ratio. Instead, we instantiate two copies of Tree A and one copy of Tree B. If the three instantiated trees receive equal bandwidth 50, then the two copies of Tree A together achieve bandwidth $50+50=100$, while Tree B receives bandwidth 50.

We now formalize this idea. Let the explicit fractional solution contain $k$ trees of positive bandwidth,
\begin{equation*}
a_1,\ldots,a_k > 0 ,
\end{equation*}
where $k$ counts only trees with $a_i>0$; trees assigned zero bandwidth by the fractional packing carry no data and are simply omitted. Fix an integer scale factor $y \ge 1$. Our goal is to create roughly $yk$ tree copies, each of which should receive a common bandwidth $X$. For a candidate value $X$, define
\begin{equation*}
m_i(X) =
\left\lfloor \frac{a_i}{X} \right\rfloor .
\end{equation*}
We create $m_i(X)$ copies of tree $T_i$. Since
\begin{equation*}
m_i(X)X \le a_i
\end{equation*}
for every $i$, the duplicated schedule with per-copy bandwidth $X$ does not consume more capacity than the original explicit solution. Therefore, if the explicit allocation is feasible, then the duplicated allocation is also feasible.

We choose the largest value $X^*$ such that the number of generated tree copies is at least $yk$:
\begin{equation}
\sum_{i=1}^k
\left\lfloor \frac{a_i}{X^*} \right\rfloor
\ge yk .
\label{eq:duplication_constraint}
\end{equation}
The choice of the \emph{largest} such $X^*$ deserves explanation, since it balances two competing pressures. The per-copy rate $X$ controls the granularity of the approximation: each tree $T_i$ is approximated by $\lfloor a_i / X \rfloor$ copies, so the emulated bandwidth $\lfloor a_i/X\rfloor X$ falls short of the target $a_i$ by the rounding remainder, which is at most $X$. A smaller $X$ makes this remainder smaller and the approximation finer, but it also increases the number of copies $\sum_i \lfloor a_i / X \rfloor$, and instantiating too many tree copies is itself costly. Conversely, a larger $X$ uses fewer copies but coarsens the approximation and risks under-approximating each $a_i$ by a larger margin. Requiring at least $yk$ copies fixes a target granularity, and among all rates achieving that granularity we take the largest $X^*$, which uses the fewest copies consistent with it. As the theorem below shows, this choice makes the total emulated bandwidth $ykX^*$ approach the explicit optimum $\sum_i a_i$ as the scale factor $y$ grows, so increasing $y$ trades more copies for a provably better approximation.

Rather than searching over all real values of $X$, we compute $X^*$ directly by sweeping the breakpoints at which copies are added. Observe that $\sum_i \lfloor a_i/X\rfloor$ is a step function of $X$ that only changes at values of the form $a_i/b$ for integers $b\ge 1$: as $X$ decreases, tree $T_i$ gains its $b$-th copy exactly when $X$ reaches $a_i/b$. Algorithm~\ref{alg:xstar} exploits this. Every tree starts with no copies, and its first breakpoint $a_i$ is placed in a max-heap; the sweep then repeatedly lowers $X$ to the next-largest breakpoint still in the heap. At each such breakpoint every tree whose breakpoint equals that value gains a copy, so the total copy count increases by at least one per step, and ties among equal breakpoints are handled automatically. The sweep stops as soon as the total first reaches or exceeds $yk$, at which point the current breakpoint is the largest $X$ satisfying Equation~\eqref{eq:duplication_constraint}, namely $X^*$. Maintaining the candidate breakpoints in a max-heap, each copy addition costs $O(\log k)$, and since at most $yk+k-1$ copies are ever created, the total cost is $O(yk\log k)$.

\begin{algorithm}[t]
\caption{Compute the duplication rate $X^*$}
\label{alg:xstar}
\begin{algorithmic}[1]
\Require Tree bandwidths $a_1,\ldots,a_k > 0$; scale factor $y \ge 1$
\Ensure Rate $X^*$ and per-tree copy counts $b_1,\ldots,b_k$
\State $b_i \gets 0$ for all $i$; \ $\mathrm{total} \gets 0$
\State \textbf{init} max-heap $H$ with keys $t_i \gets a_i$ for all $i$
       \Comment{first breakpoint of each tree}
\While{$\mathrm{total} < yk$}
    \State $X^* \gets$ maximum key in $H$
           \Comment{next breakpoint}
    \ForAll{$i$ with $t_i = X^*$}
        \State $b_i \gets b_i + 1$; \ $\mathrm{total} \gets \mathrm{total} + 1$
        \State update key of $i$ in $H$ to $t_i \gets a_i/(b_i+1)$
    \EndFor
\EndWhile
\State \Return $X^*,\ (b_1,\ldots,b_k)$
\end{algorithmic}
\end{algorithm}

Assigning bandwidth $X^*$ to every duplicated copy is a feasible allocation, since by construction the copies at per-copy rate $X^*$ consume no more capacity than the original explicit solution; it therefore serves as a feasible lower-bound allocation. By Lemma~\ref{lem:maxmin_lower_bound}, tree-coupled max-min fairness then assigns \emph{every} copy a rate at least $X^*$, so the aggregate bandwidth delivered is at least the guaranteed value $ykX^*$. We stress the distinction between this \emph{guaranteed} bandwidth $ykX^*$, which the approximation theorem below uses, and the possibly larger bandwidth that max-min fairness actually delivers: because max-min fairness raises unbottlenecked copies above $X^*$ whenever spare capacity remains, the realized aggregate may exceed $ykX^*$. The following theorem bounds the loss of the guaranteed value relative to the original explicit allocation.

\begin{theorem}[Approximation by tree duplication]
\label{thm:duplication_bound}
Let $a_1,\ldots,a_k$ be the tree bandwidths in a feasible explicit allocation, and let $X^*$ be the largest value satisfying Equation~\eqref{eq:duplication_constraint}. Then the number of instantiated tree copies satisfies
\begin{equation}
yk \le \sum_{i=1}^k \left\lfloor \frac{a_i}{X^*} \right\rfloor \le yk+k-1 < (y+1)k ,
\label{eq:copy_count_bound}
\end{equation}
and the guaranteed aggregate bandwidth $ykX^*$ satisfies
\begin{equation}
\frac{\sum_{i=1}^k a_i}{ykX^*}
\le
\frac{yk+k-1}{yk}
= 1 + \frac{k-1}{yk}.
\label{eq:duplication_ratio}
\end{equation}
\end{theorem}

\begin{proof}[Proof of Theorem~\ref{thm:duplication_bound}]
\label{proof:duplication_bound}
Let
\begin{equation*}
m_i =
\left\lfloor \frac{a_i}{X^*} \right\rfloor .
\end{equation*}
Since $X^*$ is the largest value satisfying Equation~\eqref{eq:duplication_constraint}, increasing $X^*$ by an arbitrarily small amount must make the constraint fail. That is, for every sufficiently small $\eta>0$,
\begin{equation*}
\sum_{i=1}^k
\left\lfloor \frac{a_i}{X^*+\eta} \right\rfloor
\le yk-1 .
\end{equation*}

Let $c$ be the number of indices $i$ for which $a_i/X^*$ is an integer. These are exactly the terms whose floor value decreases under an infinitesimal increase of $X^*$. Therefore,
\begin{equation*}
\sum_{i=1}^k m_i - c \le yk-1,
\end{equation*}
and hence
\begin{equation}
\sum_{i=1}^k m_i \le yk+c-1 .
\label{eq:floor_sum_bound}
\end{equation}
Since $c \le k$, this gives the upper bound $\sum_i m_i \le yk+k-1 < (y+1)k$ on the number of instantiated copies; the matching lower bound $\sum_i m_i \ge yk$ is Equation~\eqref{eq:duplication_constraint} itself. This establishes Equation~\eqref{eq:copy_count_bound}. Intuitively, if the scheme used $(y+1)k$ or more copies, then some larger value of $X$ would already yield at least $yk$ copies, contradicting the maximality of $X^*$.

We now upper-bound the explicit bandwidth. For the $c$ indices that lie exactly on an integer boundary, we have $a_i=m_iX^*$. For the remaining $k-c$ indices, we have $a_i<(m_i+1)X^*$. Therefore,
\begin{equation*}
\sum_{i=1}^k a_i
\le
X^*\sum_{i=1}^k m_i
+
X^*(k-c).
\end{equation*}
Using Equation~\eqref{eq:floor_sum_bound},
\begin{align*}
\sum_{i=1}^k a_i
&\le
X^*(yk+c-1)
+
X^*(k-c) \\
&=
X^*(yk+k-1).
\end{align*}
Dividing by $ykX^*$ gives
\begin{equation*}
\frac{\sum_{i=1}^k a_i}{ykX^*}
\le
\frac{yk+k-1}{yk}.
\end{equation*}
This proves the claimed bound.
\end{proof}

Consequently, the scheme uses at most $(y+1)k$ tree copies, and as $y$ grows the guaranteed aggregate bandwidth approaches the explicit fractional bandwidth.

The approximation factor in Theorem~\ref{thm:duplication_bound} is
\begin{equation*}
\frac{yk+k-1}{yk}
=
1+\frac{k-1}{yk}.
\end{equation*}
Thus, increasing the duplication scale $y$ makes the protocol-level schedule arbitrarily close to the explicit fractional allocation. Equivalently, for any desired accuracy parameter $\mu>0$, choosing
\begin{equation*}
y \ge \max\left\{1,\ \left\lceil\frac{k-1}{k\mu}\right\rceil\right\} \approx 1/\mu
\end{equation*}
gives an approximation factor at most $1+\mu$.

\subsection{Approximation and Time Complexity}
\label{sec:end-to-end-guarantee}
We now compose the DualColl guarantee from Section~\ref{sec:dualcoll-algorithm} with the tree-duplication
guarantee to obtain the bandwidth actually delivered under
protocol-determined max-min sharing and packet spraying, together with the total running time.

\begin{theorem}[End-to-end DualColl guarantee under packet spraying]
\label{thm:end-to-end}
Fix $0<\epsilon<1$ and run DualColl under packet spraying with accuracy
$\epsilon$, obtaining $k$ positive-rate tree pairs. Then apply tree
duplication, duplicating each pair as a unit, with the choice
\begin{equation*}
y = \max\left\{1,\ \left\lceil \frac{k-1}{k\epsilon} \right\rceil\right\} ,
\end{equation*}
i.e.\ setting the duplication accuracy parameter $\mu=\epsilon$. Then the
guaranteed aggregate bandwidth $ykX^*$ is at least
\[
(1-\epsilon)^3\,\mathrm{OPT},
\]
and the total running time is
\begin{equation}
T = \ComplexityEndToEndPS .
\label{eq:end-to-end-runtime}
\end{equation}
\end{theorem}

\begin{proof}[Proof of Theorem~\ref{thm:end-to-end}]
\label{proof:end-to-end}
By Theorem~\ref{thm:dualcoll-approximation} applied with the
packet-spraying oracle, DualColl returns a feasible explicit fractional
packing $a_1,\ldots,a_k$ (the tree pairs of positive bandwidth) with
$\sum_i a_i \ge (1-\epsilon)^2\,\mathrm{OPT}$. Tree duplication with
$y=\max\{1,\lceil (k-1)/(k\epsilon)\rceil\}$ sets the duplication accuracy to
$\mu=\epsilon$, so by Theorem~\ref{thm:duplication_bound} the guaranteed
aggregate bandwidth satisfies
\begin{align*}
ykX^*
&\ge \frac{\sum_i a_i}{1+\epsilon}
\ge (1-\epsilon)\sum_i a_i \\
&\ge (1-\epsilon)\,(1-\epsilon)^2\,\mathrm{OPT}
= (1-\epsilon)^3\,\mathrm{OPT},
\end{align*}
where the first inequality is Equation~\eqref{eq:duplication_ratio} with
$\mu=\epsilon$ and the second uses $1/(1+\epsilon)\ge 1-\epsilon$. By
Lemma~\ref{lem:maxmin_lower_bound}, tree-coupled max-min fairness delivers
at least this guaranteed value, establishing the $(1-\epsilon)^3$ bound.

For the running time we first bound the DualColl solve under packet
spraying. By Section~\ref{sec:ps-pairs} the packet-sprayed oracle
computes, on the directed projected graph, one minimum-weight spanning
in-arborescence and one minimum-weight spanning out-arborescence per
candidate root. That graph has $n$ vertices and $\Theta(n^2)$ arcs, and on
such a dense graph Tarjan's implementation of the Chu--Liu/Edmonds
algorithm computes a minimum spanning arborescence in $O(n^2)$
time~\cite{tarjan1977branchings}, so the $2n$ arborescence
computations cost $\ComplexityPSArborescences$. Adding the projection cost
$\ComplexityProjectionDPAll$ of Equation~\eqref{eq:ps-projection-cost},
one oracle invocation costs
\begin{equation}
T_{\mathrm{oracle}}^{\mathrm{PS}} = \ComplexityOraclePS ,
\label{eq:ps-oracle-runtime}
\end{equation}
and since the Garg--K\"onemann scheme performs
$N_{\mathrm{iter}} = \ComplexityNiter$ iterations, each invoking the
oracle once, the solve costs
\begin{equation}
T_{\mathrm{DualColl}}^{\mathrm{PS}} = \ComplexityTotalPS .
\label{eq:dualcoll-ps-runtime}
\end{equation}
Computing
$X^*$ and instantiating the copies via Algorithm~\ref{alg:xstar} costs
$O(yk\log k)$. With $y=\max\{1,\lceil (k-1)/(k\epsilon)\rceil\}$ we have
$yk = O(k/\epsilon)$, so the duplication step costs $O((k/\epsilon)\log k)$.
The number of distinct tree pairs returned by the packing algorithm is at most
the number of iterations, $k = O(N_{\mathrm{iter}}) = O(|E|\log|E|/\epsilon^2)$;
substituting this bound and using
$\log k = O(\log(|E|/\epsilon))$ gives a duplication cost of
$O\!\left(|E|\log|E|\,\log(|E|/\epsilon)/\epsilon^3\right)$. The total is the
sum of the DualColl solve and this duplication cost, which is
Equation~\eqref{eq:end-to-end-runtime}.
\end{proof}

The two terms in Equation~\eqref{eq:end-to-end-runtime} do not dominate
one another in general: the DualColl solve scales as $1/\epsilon^2$ while
the duplication term scales as $1/\epsilon^3$. For fixed network size and
$\epsilon\to 0$ the duplication term eventually dominates, while for
moderate $\epsilon$ or large networks the DualColl solve dominates; the
total is therefore stated as their sum.

\section{Extensions to Other Collective Operations}
\label{sec:other-collectives}

This section extends DualColl beyond a single AllReduce: to AllGather, to
ReduceScatter, and to multiple concurrent AllReduce groups. In each case
only the packing variable and the oracle change; the multiplicative-weights
framework of Section~\ref{sec:dualcoll-algorithm} is used unchanged.

\subsection{The AllGather Operation}
\label{subsection:allgather}

The AllGather and ReduceScatter constructions below, and the bottleneck
objective for multiple concurrent AllReduce groups in
Section~\ref{subsection:multiset}, all rest on the same reduction from a
max-min objective to a tuple packing. We first motivate it with AllGather
and then state it once, in general form.

In the AllGather operation, every node starts with a unique data buffer and must broadcast it to all other nodes in the network. For each node $v_j$, there exists at least one distinct tree rooted at $v_j$ that spans all other nodes to broadcast its data. Let $\mathcal{T}_j$ represent the set of all trees rooted at node $v_j$. The logical structure for a specific root must strictly be a tree (containing no cycles) to avoid redundant data transmission. Each directed logical edge maps to a single unidirectional flow in the physical topology.

Every node must broadcast its full data buffer of size $X$ to every other node. Write $B_i$ for the bandwidth of tree $T_i$. The total bandwidth actively broadcasting node $v_j$'s data is the sum of the bandwidths of the trees rooted at $v_j$, denoted as $\Lambda_j = \sum_{T_i \in \mathcal{T}_j} B_i$. The time for node $v_j$ to finish its broadcast is $X / \Lambda_j$. The completion time is determined by the root with the smallest aggregate broadcast rate:
\begin{equation}
    \text{Total Time} = \max_j \frac{X}{\Lambda_j} = \frac{X}{\min_j \Lambda_j}
\end{equation}
Hence the objective is to maximize the minimum aggregate broadcast rate across roots, $\max \min_j \Lambda_j$.

\paragraph{From a max-min objective to a tuple packing.}
Suppose the trees of interest are partitioned into $k$ \emph{classes}
$\mathcal{T}_1,\ldots,\mathcal{T}_k$, all drawing on the same physical-link
capacities. Tree $T_i$ has bandwidth $B_i$, and the aggregate bandwidth of
class $j$ is $\Lambda_j=\sum_{T_i\in\mathcal{T}_j}B_i$. A bandwidth vector
$B$ is \emph{feasible} if $\sum_i A_{e,T_i}B_i\le C_e$ for every directed
edge $e$, and the objective is to maximize $\min_j \Lambda_j$. The
following lemma says that this max-min problem has the same optimum as an
ordinary fractional packing LP whose columns are \emph{tuples} containing
one tree from each class, all trees of a tuple sharing a single common
rate.

\begin{lemma}[Tuple packing for max-min objectives]
\label{lem:tuple-packing}
Let $\mathcal{C}$ be the set of tuples $c=(T^{(1)},\ldots,T^{(k)})$ with
$T^{(j)}\in\mathcal{T}_j$, and for a tuple $c$ let
$A_{e,c}=\sum_j A_{e,T^{(j)}}$ be its total load on directed edge $e$.
Then the optimum of the max-min problem
\[
\max\Bigl\{\min_j \Lambda_j \;:\; B \text{ feasible}\Bigr\}
\]
equals the optimum of the tuple-packing LP
\[
\max\Bigl\{\textstyle\sum_{c\in\mathcal{C}}\lambda_c \;:\;
\sum_{c}A_{e,c}\lambda_c\le C_e\ \forall e,\ \lambda\ge0\Bigr\}.
\]
\end{lemma}

\begin{proof}
We show each side is at least the other.

\emph{Tuple packing $\le$ max-min.} Let $\lambda$ be a feasible tuple
packing of value $z=\sum_c \lambda_c$. Assign every tree appearing in
tuple $c$ the rate $\lambda_c$, summing over tuples when a tree occurs in
several. Each class then receives aggregate bandwidth exactly
$\Lambda_j=\sum_c\lambda_c=z$, because every tuple contributes one tree to
every class. The induced per-edge load is
$\sum_c\lambda_c A_{e,c}=\sum_c\lambda_c\sum_j A_{e,T^{(j)}}$, which is
exactly the load of the tuple packing and hence at most $C_e$. So this is
a feasible solution of the max-min problem with $\min_j \Lambda_j = z$.

\emph{Max-min $\le$ tuple packing.} Let $B$ be feasible with
$z=\min_j \Lambda_j$. We start by lowering the rates of the trees of any class
whose aggregate exceeds $z$ until every class has aggregate exactly $z$, and the result will still be feasible. Now build tuples greedily. While some
tree has positive rate, let $\delta$ be the \emph{smallest} positive rate
among all trees of all classes, attained by a tree $T_{i^\star}$; pick any
positive-rate tree from each class, including $T_{i^\star}$ from its own,
and form the tuple $c$ of these trees with rate $\lambda_c=\delta$.
Subtract $\delta$ from the rate of each chosen tree; by the choice of
$\delta$ no rate becomes negative. Each step removes exactly $\delta$ from
every class, so all classes always have the same remaining aggregate.
Consequently, as long as some class still has a tree of positive rate,
every class does: we never reach a situation in which one class has no
tree left while another still has one, and each tuple can always be
completed. The tree $T_{i^\star}$ now has rate zero, so the number of
positive-rate trees strictly decreases and the process terminates, with
all classes reaching zero together. The tuple rates sum to $z$, and every
tree's rate is fully distributed among the tuples containing it, so the
per-edge load of the resulting tuple packing equals that of $B$ and is at
most $C_e$. Its value is therefore $z$.
\end{proof}

Lemma~\ref{lem:tuple-packing} is applied with one class per AllGather root
below, with one class per ReduceScatter root in
Section~\ref{subsection:reducescatter}, and with one class per concurrent
AllReduce group in Section~\ref{subsection:multiset}.

\paragraph{Our Solution.}
By Lemma~\ref{lem:tuple-packing}, applied with one class per root, the
max-min objective above has the same optimum as a fractional packing whose
columns are $n$-tuples of trees, one rooted at each compute node, with all
trees of a tuple sharing a common rate. We therefore take a packing
variable to be such an $n$-tuple, and the oracle must return the
minimum-cost one under the current dual prices. Since the cost of a tuple
is the sum of the costs of its $n$ trees and the trees are chosen
independently, the cheapest tuple consists of the cheapest tree rooted at
each node, so the oracle computes, for every root, a minimum-weight
spanning out-arborescence on the projected graph.

\paragraph{Complexity of one oracle call.}
The projected graph $G_R$ has $n$ vertices, one per participating
rank, and is complete and directed, so it contains
$\Theta(n^2)$ arcs. On such a dense graph Tarjan's implementation of the
Chu--Liu/Edmonds algorithm computes a minimum spanning arborescence in
$O(n^2)$ time~\cite{tarjan1977branchings}. The oracle requires one such
arborescence for each of the $n$ roots, giving a total arborescence cost
of $O(n^3)$.

The projected graph itself is constructed once per oracle call.
Under scheduler-routed single-path routing, its directed edge
weights are shortest-path distances under the current dual prices.
Running Dijkstra's algorithm from every participating rank with a
binary heap requires
\[
O\!\left(n(N+|E|)\log N\right)
\]
time. Under local equal-next-hop packet spraying, the directed
projection weights are the expected dual costs described in
Section~\ref{sec:proj}; one destination-rooted sweep per
participating rank computes all such weights in
\[
O\!\left(n(N+|E|)\right)
\]
time.

Hence, under scheduler-routed routing, one AllGather oracle call
takes
\[
T_{\mathrm{oracle}}^{\mathrm{AG,SP}}
=
O\!\left(
n(N+|E|)\log N+n^3
\right),
\]
whereas under packet spraying it takes
\[
T_{\mathrm{oracle}}^{\mathrm{AG,PS}}
=
O\!\left(
n(N+|E|)+n^3
\right).
\]

\paragraph{Total complexity.}
The Garg--K\"onemann framework performs
\[
O\!\left(
\frac{|E|\log |E|}{\epsilon^2}
\right)
\]
oracle iterations. Therefore, the total running time for
AllGather under scheduler-routed routing is
\[
T_{\mathrm{AG,SP}}
=
O\!\left(
\frac{|E|\log |E|}{\epsilon^2}
\left[
n(N+|E|)\log N+n^3
\right]
\right),
\]
and under packet spraying it is
\[
T_{\mathrm{AG,PS}}
=
O\!\left(
\frac{|E|\log |E|}{\epsilon^2}
\left[
n(N+|E|)+n^3
\right]
\right).
\]

In common fabrics where $N=O(n)$ and hence
$|E|=O(n^2)$, the $O(n^3)$ arborescence term is of the same
order as, or may dominate, the projection term depending on the
physical topology. We therefore retain both terms in the general
bounds above.

\subsection{The ReduceScatter Operation}
\label{subsection:reducescatter}

In the ReduceScatter operation, every node starts with a full input vector and ends with a specific reduced block of the data. As in AllGather, there is a set of trees $\mathcal{T}_j$ associated with each node $v_j$; here, however, the edges in each tree are directed \emph{towards} the root $v_j$, representing the flow of partial results being aggregated from the leaves to the root. The completion time is again governed by the node whose aggregation is slowest, so the objective is the same max-min objective as in AllGather, $\max \min_j \Lambda_j$ with $\Lambda_j=\sum_{T_i \in \mathcal{T}_j} B_i$.

\paragraph{Why we do not obtain ReduceScatter by reversal.}
ReduceScatter is commonly treated as the directional dual of AllGather
and obtained by reversing every logical edge of an AllGather solution;
this is what \textit{ForestColl} does, and under scheduler-routed
single-path routing it is sound. When every logical edge is mapped to one
physical path, reversing that path preserves feasibility whenever the two
directions of each physical link have equal capacity, since the reversed
flow consumes on $\bar e$ exactly what the original consumed on $e$.

Under local equal next-hop packet spraying this argument no longer
applies. The fractions $p_e(F)$ induced by a flow $F=(u,v)$ are determined
by the shortest-path DAG toward $v$ and by the branching degrees
$|\Gamma_v(c)|$ along it. The reversed flow $(v,u)$ is routed on the
shortest-path DAG toward $u$, whose branching structure may differ, so in
general
\[
p_{\bar e}\bigl((v,u)\bigr) \neq p_e\bigl((u,v)\bigr),
\]
even on a symmetric full-duplex topology. Reversal therefore need not
preserve the induced routing fractions, and hence need not preserve
feasibility or per-tree bandwidth. We consequently do not derive
ReduceScatter from AllGather by reversal, and instead solve it directly.

\paragraph{Our Solution.}
By Lemma~\ref{lem:tuple-packing}, applied with one class per root, this
max-min objective again has the same optimum as a fractional packing over
$n$-tuples of trees. Solving ReduceScatter directly therefore requires only
one change to the AllGather oracle of Section~\ref{subsection:allgather}. There, a packing variable is
an $n$-tuple of trees, one rooted at each node, and the oracle returns for
each root a minimum-weight \emph{out}-arborescence, whose edges are
directed away from the root, matching the outward flow of a broadcast. For
ReduceScatter the data flows inward instead, so the oracle returns for
each root a minimum-weight \emph{in}-arborescence, whose edges are directed
toward the root; equivalently, one runs the same minimum-arborescence
computation on the projected graph with its edge orientations reversed.
Everything else is unchanged: the projected weights are computed from the
current dual prices exactly as before, each tuple receives a common
bandwidth, and the surrounding multiplicative-weights framework, its
stopping rule and its scaling step are untouched. Because the two oracles
differ only in the orientation of the arborescences they compute, and an
in-arborescence on $n$ nodes costs the same as an out-arborescence,
ReduceScatter inherits the AllGather running time and approximation
guarantee, and the argument is valid under packet-sprayed routing as well
as under single-path routing.

\subsection{Multiple Concurrent AllReduce-Sets}
\label{subsection:multiset}
In complex distributed systems, an AllReduce-set is a specific subset of compute nodes participating in an independent AllReduce operation. We consider two distinct optimization objectives for this scenario: maximizing the total aggregate bandwidth across all concurrent sets, or maximizing the bottleneck (slowest) AllReduce-set.

\paragraph{Our Solution.} An important assumption required for optimality in this context is that compute nodes only participate in the routing and computation of their own AllReduce-sets. Therefore, a spanning tree for a given set contains only and all the nodes belonging to that specific set.

To maximize the total aggregate bandwidth across all sets, we utilize the original fractional packing algorithm (Section~\ref{sec:dualcoll-algorithm}). During each oracle iteration, the oracle computes the MST for \textit{each} active AllReduce-set independently, and strictly returns the single tree (from any set) with the absolute minimum weight.

Alternatively, to maximize the slowest AllReduce-set, we apply Lemma~\ref{lem:tuple-packing} with one class per concurrent set: the bottleneck objective then has the same optimum as a fractional packing whose columns are $k$-tuples of trees, one spanning each AllReduce-set, all sharing a common rate. To solve this, we structure our oracle to compute independent spanning trees per set during each iteration (utilizing only the compute nodes relevant to that specific set). The single column vector returned to the solver is exactly the union of these independently optimal spanning trees. This increases the oracle complexity by a factor of $k$, giving $\ComplexityMultiSet$, and makes the framework directly applicable to multi-tenant environments. Under packet spraying, each tree above becomes a reduce/broadcast tree pair, and each minimum spanning tree computation is replaced by the pair oracle of Section~\ref{sec:ps-pairs}.

\section{Related Work}
\label{sec:related-work}
\paragraph{Collective communication libraries and algorithms.}
High-performance distributed systems have long relied on optimized
collective-communication algorithms. Classical MPI work developed
bandwidth-efficient algorithms for operations such as AllReduce and
ReduceScatter~\cite{rabenseifner2004,thakur2005,patarasuk2009}, while
modern distributed-training systems commonly use GPU-oriented libraries
such as NVIDIA NCCL~\cite{nccl} and frameworks such as
Horovod~\cite{horovod}. BlueConnect adapts AllReduce to heterogeneous
network hierarchies by decomposing it into parallel ReduceScatter and
AllGather operations that can be mapped to different network
fabrics~\cite{blueconnect}. More recent systems provide increasingly
programmable communication layers; for example, MSCCLang provides a
language for writing custom GPU collective algorithms together with an
optimizing compiler~\cite{mscclang}, while MSCCL++ provides low-level
communication primitives and a DSL for constructing high-performance
GPU communication kernels~\cite{mscclpp}. DualColl is complementary to
these systems: our focus is the algorithmic construction of provably
high-throughput schedules for a given heterogeneous network fabric.

\paragraph{Topology-aware synthesis and tree packing.}
Several systems automatically construct collective algorithms for a
given topology. SCCL formulates collective synthesis as an SMT problem
and searches for algorithms along the latency--bandwidth
tradeoff~\cite{sccl}, while TACCL uses user-provided communication
sketches to guide synthesis on larger heterogeneous
topologies~\cite{taccl}. Blink~\cite{blink} is particularly related to
DualColl: it applies a multiplicative-weight approximation algorithm to
pack spanning trees on heterogeneous GPU topologies and uses the
resulting structures to implement collectives including AllReduce.
DualColl differs in the optimization problem and guarantees considered
here: we formulate bandwidth-bound AllReduce directly, prove that the
tree-based restriction loses no throughput under our full-duplex model,
and obtain a polynomial-time $(1-\epsilon)^2$ approximation guarantee
for the resulting AllReduce problem. Blink's tree-packing formulation
does not directly support switched topologies.

\paragraph{ForestColl.}
The work most closely related to ours is
\emph{ForestColl}~\cite{forestcoll}, which constructs throughput-optimal
collective schedules on heterogeneous network fabrics by packing
broadcast or aggregation spanning trees. Its optimality characterization
is based on the network bottleneck cut and also enables switch nodes to
be eliminated while preserving the relevant cut capacities. ForestColl's
polynomial construction is developed for AllGather and ReduceScatter,
with AllReduce obtained by composing the two. We show that this
composition can be strictly suboptimal for AllReduce
(Appendix~\ref{app:counterexample}) and instead optimize AllReduce
directly. Our treatment of switches also differs: rather than eliminating
them using a cut transformation, DualColl projects physical-link dual
prices onto a logical graph over participating ranks; switch links remain
fully represented through the projected edge costs and capacity
constraints. Finally, whereas ForestColl considers scheduler-routed
single-path routing and a single collective group, DualColl additionally
addresses packet-sprayed routing, tree-coupled max-min bandwidth
allocation, and multiple concurrent collective groups.

\section{Conclusion}

\label{sec:conclusion}
We formulated the scheduling of bandwidth-bound collective operations as a fractional tree-packing LP under an explicit bandwidth-allocation model. By combining the Garg--K\"onemann algorithm with an exact dual-cost projection oracle, we obtain a fully polynomial-time approximation scheme (FPTAS) that computes a solution arbitrarily close to the optimal fractional tree-packing value on any static datacenter topology satisfying our
network model. To the best of our knowledge, DualColl is the first polynomial-time
algorithm with a provable near-optimality guarantee for bandwidth-bound
AllReduce on switched topologies.

We further extend the framework beyond the scheduler-routed setting: to
\emph{local equal next-hop packet spraying}, where the reduce and
broadcast trees must be packed as pairs; to networks without explicit
per-tree bandwidth control, where tree duplication realizes a fractional
allocation under tree-coupled max-min fairness; and to the other
collective operations and to multiple concurrent collective groups.

\section*{Acknowledgments}

\paragraph{Use of Generative AI} Generative AI tools (Claude, Anthropic;
and ChatGPT, OpenAI) were used during the preparation of this paper for
language editing and technical review. All research ideas, models,
algorithms, theorem statements, proofs, and technical claims were
determined by the authors, who take full responsibility for the content.

\bibliographystyle{ACM-Reference-Format}
\bibliography{DualColl}

\appendix

\section{Sub-Optimality of ReduceScatter+AllGather}
\label{app:counterexample}

This appendix exhibits an instance on which an AllGather-optimal set of
trees, composed into an AllReduce in the usual way, is strictly worse than
the optimal AllReduce. It shows that optimality of the AllGather phase does
not certify optimality of the composed AllReduce.

\begin{figure*}[tbp]
\centering
\begin{subfigure}[t]{0.46\textwidth}
  \centering
  \begin{tikzpicture}[scale=0.7, transform shape]
    \node[circle, draw, minimum size=0.6cm] (0) at (0, 0) {0};
    \node[circle, draw, minimum size=0.6cm] (1) at (2, 0) {1};
    \node[circle, draw, minimum size=0.6cm] (2) at (4, 0) {2};
    \node[circle, draw, minimum size=0.6cm] (3) at (6, 0) {3};
    \node[circle, draw, minimum size=0.6cm] (4) at (8, 0) {4};
    \draw[thick] (0) -- node[above] {\small 4} (1);
    \draw[thick] (1) -- node[above] {\small 4} (2);
    \draw[thick] (2) -- node[above] {\small 4} (3);
    \draw[thick] (3) -- node[above] {\small 4} (4);
  \end{tikzpicture}
  \caption{Topology: five ranks on a path, every link of capacity~4 in
  each direction.}
  \label{fig:ce-topology}
\end{subfigure}
\hfill
\begin{subfigure}[t]{0.46\textwidth}
  \centering
  \begin{tikzpicture}[scale=0.7, transform shape]
    \node[circle, draw, minimum size=0.6cm] (0) at (0, 0) {0};
    \node[circle, draw, minimum size=0.6cm] (1) at (2, 0) {1};
    \node[circle, draw, minimum size=0.6cm] (2) at (4, 0) {2};
    \node[circle, draw, minimum size=0.6cm] (3) at (6, 0) {3};
    \node[circle, draw, minimum size=0.6cm] (4) at (8, 0) {4};
    \draw[thick, gray!40] (0) -- (1);
    \draw[thick, gray!40] (1) -- (2);
    \draw[thick, gray!40] (2) -- (3);
    \draw[thick, gray!40] (3) -- (4);
    \draw[very thick, blue, ->, bend left=45] (0) to (2);
    \draw[very thick, blue, ->] (2) -- (1);
    \draw[very thick, blue, ->] (2) -- (3);
    \draw[very thick, blue, ->] (3) -- (4);
  \end{tikzpicture}
  \caption{The AllGather tree $T_0$: the logical edge $0\!\to\!2$ is
  routed over $0$--$1$--$2$, then $2\!\to\!1$, $2\!\to\!3$, $3\!\to\!4$.}
  \label{fig:ce-t0}
\end{subfigure}

\vspace{0.6em}

\begin{subfigure}[t]{0.46\textwidth}
  \centering
  \begin{tikzpicture}[scale=0.7, transform shape]
    \node[circle, draw, minimum size=0.6cm] (0) at (0, 0) {0};
    \node[circle, draw, minimum size=0.6cm] (1) at (2, 0) {1};
    \node[circle, draw, minimum size=0.6cm] (2) at (4, 0) {2};
    \node[circle, draw, minimum size=0.6cm] (3) at (6, 0) {3};
    \node[circle, draw, minimum size=0.6cm] (4) at (8, 0) {4};
    \draw[thick, gray!40] (0) -- (1);
    \draw[thick, gray!40] (1) -- (2);
    \draw[thick, gray!40] (2) -- (3);
    \draw[thick, gray!40] (3) -- (4);
    \draw[very thick, red, ->] (3) -- (2);
    \draw[very thick, red, ->] (2) -- (1);
    \draw[very thick, red, ->] (1) -- (0);
    \draw[very thick, red, ->] (3) -- (4);
  \end{tikzpicture}
  \caption{A typical AllGather tree, here $T_3$: data spreads outward
  from the root along the path.}
  \label{fig:ce-t3}
\end{subfigure}
\hfill
\begin{subfigure}[t]{0.46\textwidth}
  \centering
  \begin{tikzpicture}[scale=0.7, transform shape]
    \node[circle, draw, minimum size=0.6cm] (0) at (0, 0) {0};
    \node[circle, draw, minimum size=0.6cm] (1) at (2, 0) {1};
    \node[circle, draw, minimum size=0.6cm] (2) at (4, 0) {2};
    \node[circle, draw, minimum size=0.6cm] (3) at (6, 0) {3};
    \node[circle, draw, minimum size=0.6cm] (4) at (8, 0) {4};
    \draw[very thick, teal] (0) -- (1);
    \draw[very thick, teal] (1) -- (2);
    \draw[very thick, teal] (2) -- (3);
    \draw[very thick, teal] (3) -- (4);
  \end{tikzpicture}
  \caption{The single AllReduce tree: the path itself, reducing toward a
  root and broadcasting back along the same edges. It attains
  bandwidth~4.}
  \label{fig:ce-direct}
\end{subfigure}
\caption{A five-rank path on which an AllGather-optimal tree set composes
into a strictly suboptimal AllReduce. The AllGather trees of
Figures~\ref{fig:ce-t0} and~\ref{fig:ce-t3} all run at bandwidth~1, which
is AllGather-optimal, but converting them into AllReduce trees overloads
the link $1$--$2$ and caps the aggregate at $10/3$. The single tree of
Figure~\ref{fig:ce-direct} achieves~$4$.}
\label{fig:counterexample}
\Description{Four diagrams of five nodes numbered 0 to 4 joined in a line,
every link of capacity 4. The first shows the bare topology. The second
shows the AllGather tree rooted at node 0, whose logical edge to node 2 is
routed through node 1, after which node 2 feeds nodes 1 and 3 and node 3
feeds node 4. The third shows a typical AllGather tree spreading outward
from node 3. The fourth shows the single AllReduce tree consisting of the
whole line.}
\end{figure*}
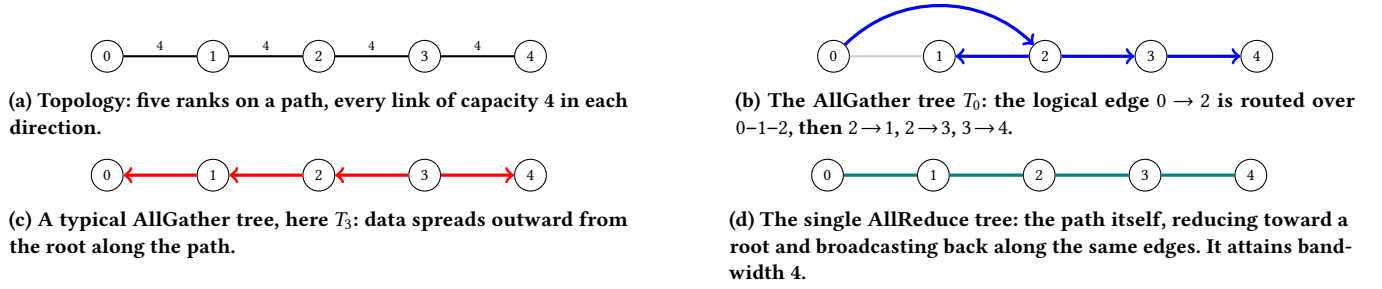

\paragraph{The instance.}
Consider five ranks $0,1,2,3,4$ connected in a path, with links
$\{0,1\},\{1,2\},\{2,3\},\{3,4\}$, each of capacity $4$ in each direction
(Figure~\ref{fig:ce-topology}).

\paragraph{An optimal AllGather solution.}
For AllGather, rank $4$ must receive the data of all four other ranks, and
every such stream crosses the directed edge $3\!\to\!4$. If all trees are
run at a common bandwidth $\beta$, this forces $4\beta \le 4$, so
$\beta \le 1$; a solution with $\beta = 1$ is therefore AllGather-optimal.
Consider the following set of broadcast trees, one rooted at each rank,
all run at $\beta = 1$:
\begin{align*}
T_0 &: \ 0\!\to\!2,\ 2\!\to\!1,\ 2\!\to\!3,\ 3\!\to\!4, \\
T_1 &: \ 1\!\to\!0,\ 1\!\to\!2,\ 2\!\to\!3,\ 3\!\to\!4, \\
T_2 &: \ 2\!\to\!1,\ 1\!\to\!0,\ 2\!\to\!3,\ 3\!\to\!4, \\
T_3 &: \ 3\!\to\!2,\ 2\!\to\!1,\ 1\!\to\!0,\ 3\!\to\!4, \\
T_4 &: \ 4\!\to\!3,\ 3\!\to\!2,\ 2\!\to\!1,\ 1\!\to\!0 .
\end{align*}
All trees except $T_0$ simply spread outward along the path. The tree
$T_0$ is the twist: instead of handing its data to rank $1$ first, rank
$0$ sends a single logical flow to rank $2$, routed over the physical path
$0$--$1$--$2$, and rank $2$ then serves ranks $1$ and $3$
(Figure~\ref{fig:ce-t0}). Counting, for each directed physical edge, the
routed logical tree edges that traverse it gives the loads
\[
\begin{array}{llll}
0\!\to\!1: 1, & 1\!\to\!0: 4, & 1\!\to\!2: 2, & 2\!\to\!1: 4, \\
2\!\to\!3: 3, & 3\!\to\!2: 2, & 3\!\to\!4: 4, & 4\!\to\!3: 1 .
\end{array}
\]
Every load is at most the capacity $4$, so this tree set is feasible, and
since $\beta = 1$ it is AllGather-optimal.

\paragraph{The composed AllReduce.}
An AllReduce built from this solution runs a ReduceScatter phase that is
the reversal of the AllGather phase, so in steady state each tree uses both
directions of the links it occupies: the load it places on a directed edge
$e$ is its AllGather load on $e$ plus its AllGather load on the reverse
edge $\bar e$. Adding the table above to its own reversal gives, per unit
of common tree bandwidth,
\[
\begin{array}{llll}
0\!\to\!1: 5, & 1\!\to\!0: 5, & 1\!\to\!2: 6, & 2\!\to\!1: 6, \\
2\!\to\!3: 5, & 3\!\to\!2: 5, & 3\!\to\!4: 5, & 4\!\to\!3: 5 .
\end{array}
\]
The link $\{1,2\}$ now carries $6$ units of load per unit of bandwidth in
both directions. The reason is that the directed edge $1\!\to\!2$ already
carried $2$ units in the AllGather phase: one from $T_1$, whose logical
edge $1\!\to\!2$ uses it directly, and one from the twisted tree $T_0$,
whose logical edge $0\!\to\!2$ is routed over the physical path
$0$--$1$--$2$ and therefore traverses $1\!\to\!2$ as well. The opposite
direction $2\!\to\!1$ carried $4$ units, one from each of $T_0, T_2, T_3,
T_4$. Reversing the trees for the ReduceScatter phase superimposes these
two figures on each direction, giving $2+4=6$. Feasibility therefore requires
$6\beta \le 4$, that is $\beta \le 2/3$, and the aggregate AllReduce
bandwidth of the five trees is at most
\[
5 \cdot \tfrac{2}{3} \;=\; \tfrac{10}{3} \;\approx\; 3.33 .
\]

\paragraph{A better direct AllReduce.}
Now take the single tree consisting of the path itself, reducing toward a
root and broadcasting back along the same edges
(Figure~\ref{fig:ce-direct}). Each directed physical edge carries exactly
one logical tree edge, so its load is $1$ per unit of tree bandwidth, and
the tree may be run at bandwidth $4$. The aggregate AllReduce bandwidth is
therefore $4 > 10/3$.

\paragraph{Conclusion.}
The tree set above is optimal for AllGather, yet the AllReduce composed
from it attains at most $10/3$, while a direct AllReduce schedule on the
same network attains $4$. The composed AllReduce is therefore suboptimal.

\section{Comparison with ForestColl's Complexity}
\label{app:comparison}

The authors of \textit{ForestColl} did not solve AllReduce directly using their main polynomial algorithm; instead, they solved AllGather and used a combination of AllGather and ReduceScatter to emulate the AllReduce operation. While they do supply a linear program for solving AllReduce directly, it is unscalable. 

For their polynomial AllGather-based AllReduce algorithm, they establish an asymptotic time complexity of $\ComplexityForestColl$. \textit{ForestColl} operates under the same scheduler-routed model that we present as our primary model, and supports only a single collective group. We therefore compare against our direct Section~\ref{sec:dualcoll-algorithm} algorithm: solving AllReduce directly via the Garg--K\"onemann framework, our oracle runs Dijkstra from every rank and one MST per iteration, giving an overall time complexity of $\ComplexityTotal$. Since the number of physical links $|E|$ is at most $O(N^2)$, our algorithm offers a significantly tighter and more scalable polynomial bound, making it substantially more practical for large-scale datacenter networks.

The comparison is not limited to AllReduce. Our framework also solves
AllGather, the operation for which \textit{ForestColl} does establish
optimality, and does so with a better bound. Replacing the minimum
spanning tree in our oracle by a minimum spanning arborescence rooted at
each compute node (Section~\ref{subsection:allgather}) yields an
AllGather schedule in time $\ComplexityAllgather$, against the
$\ComplexityForestColl$ of \textit{ForestColl}. When $N=O(n)$, so that $|E| = O(n^2)$, our
AllGather bound is $O(n^5/\epsilon^2)$ up to logarithmic factors,
a polynomial factor of roughly $n^3$ below theirs. Thus DualColl improves
on \textit{ForestColl} both where \textit{ForestColl} is provably optimal,
by computing an AllGather schedule faster, and where it is not, by giving
AllReduce a guarantee that the composed construction does not provide.

\end{document}